%% file: main.tex
\documentclass[a4paper]{llncs}
\usepackage[notes=true,later=true,llncssubsub]{dtrt}
\usepackage[advantage,adversary,asymptotics,complexity,keys,lambda,landau,logic,sets,mm,oracles,primitives,probability,notions,sets,operators]{cryptocode} 
\usepackage{color, colortbl, tcolorbox}
\usepackage{xurl, array, cite, multirow, mathtools, framed, newfloat,microtype,algpseudocode,algorithm}
\usepackage[T1]{fontenc}
\usepackage{diagbox, booktabs, appendix, alltt}
\usepackage{listings, xcolor, afterpage, bm, xspace}
\usepackage[inkscapelatex=false]{svg}
\usepackage[shortlabels]{enumitem}
\usepackage[multiple]{footmisc}
\usepackage[cal=boondox]{mathalfa}
\usepackage[capitalise]{cleveref}
\usepackage{stmaryrd}
\usepackage{caption}
\usepackage{tcolorbox}
\usepackage[english]{babel}
\usepackage{comment}
\usepackage{subcaption}
\usepackage[mathscr]{euscript}
\usepackage{tabularx}
\usepackage{tikz}
\usetikzlibrary{arrows.meta, positioning, fit}
\usepackage{centernot}

\usepackage{soul}
\usepackage[strings]{underscore}

\input{macros}

\begin{document}

\title{Adversary Resilient Learned Bloom Filters}

\newif\ifanonymoussubmission
\anonymoussubmissionfalse

\ifanonymoussubmission
    \author{Anonymous Authors}
\else
    \author{Ghada Almashaqbeh \inst{1} \and Allison Bishop\inst{2,3} \and
    Hayder Tirmazi\inst{3}}

    \institute{ University of Connecticut, \email{ghada@uconn.edu}
    \\
    \and Proof Trading, \email{abishop@ccny.cuny.edu}
    \\
    \and City College of New York, \email{hayder.research@gmail.com} }
\fi

\maketitle              
\input{tex/abstract}
\input{tex/intro}

\input{tex/prelim}
\input{tex/model}
\input{tex/construction}
\input{tex/sec-perf-analysis}
\input{tex/perf-evaluation}

\input{tex/acknowledgments}

\bibliographystyle{splncs04}
\bibliography{paper}

\appendix
\input{tex/appendices}

\end{document}

%% file: macros.tex
\DeclareMathAlphabet{\mathcal}{OMS}{cmsy}{m}{n}

\spnewtheorem{challenge}{Game}{\bfseries}{}
\spnewtheorem{construction}{Construction}{\bfseries}{}

\spnewtheorem{lemma}{Lemma}{\bfseries}{}
\spnewtheorem{theorem}{Theorem}{\bfseries}{}
\spnewtheorem{corollary}{Corollary}{\bfseries}{}

\graphicspath{{content}}

\newcommand{\bloomname}[0]{PRP-LBF}

\newcommand{\betpassbloomname}[0]{Cuckoo-LBF}

\newcommand{\securecbf}[0]{NY-CBF}
\newcommand{\betpasscbf}[0]{NOY-Cuckoo Filter}

\newcommand{\domain}{\mathfrak{D}}

\let\olddefinition\definition
\renewcommand{\definition}{\olddefinition\normalfont}

\let\oldchallenge\challenge
\renewcommand{\challenge}{\oldchallenge\normalfont}

%% file: tex/abstract.tex
\begin{abstract}
A learned Bloom filter (LBF) combines a classical Bloom filter (CBF) with a learning model to reduce the amount of memory needed to represent a given set while achieving a target false positive rate (FPR). Provable security against adaptive adversaries that advertently attempt to increase FPR has been studied for CBFs, but not for LBFs. In this paper, we close this gap and show how to achieve adaptive security for LBFs. In particular, we define several adaptive security notions capturing varying degrees of adversarial control, including full and partial adaptivity, in addition to LBF extensions of existing adversarial models for CBFs, including the Always-Bet and Bet-or-Pass notions. We propose two secure LBF constructions, \bloomname{} and \betpassbloomname{}, and formally prove their security under these models assuming the existence of one-way functions. Based on our analysis and use case evaluations, our constructions achieve strong security guarantees while maintaining competitive FPR and memory overhead.
\end{abstract}

%% file: tex/intro.tex
\section{Introduction}
\label{sec:intro}
Bloom filters are probabilistic data structures that allow building a succinct representation of a data set while offering approximate membership queries, i.e., whether an element $x$ is in a set $S$. Bloom filters have one-sided error guarantees: if $x \in S$ the filter always returns yes (i.e., no false negatives), but if $x \notin S$, the filter may return yes, instead of no, with some probability resulting in a false positive. Bloom filters greatly improve space efficiency; instead of storing the full set, this set is encoded in a much shorter bit string using a family of hash functions to map each item into a few bits (set to 1) in this string~\cite{BroderMitzenmacher2004}.

Learned Bloom filters (LBFs) have been introduced to improve upon conventional, or classical Bloom filters (CBFs)~\cite{learnedindexstructures}. In particular, and as shown in Figure~\ref{fig:learnedbloom}, an LBF combines a learning model with a CBF to obtain a lower false positive rate (FPR) than CBFs under the same memory budget. The learning model acts as an initial filter, providing a probabilistic estimate on whether $x \in S$, while a smaller CBF serves as a backup to prevent any false negatives. 

Bloom filters are used in various practical applications, including Linux network drivers, network protocols, deep packet scanners, peer-to-peer networking, and caching~\cite{tarkoma1,BroderMitzenmacher2004}. Many well-known systems, such as Apache Hadoop, Apache HBase, Google BigTable, Google LevelDB, and Meta RocksDB, use Bloom filters as part of their implementations. As LBFs provide the same one-sided error guarantees as CBFs, the applications of LBFs and CBFs are identical. Recently, Roblox used LBFs to achieve 25\% cost savings in their production workloads for Spark join queries~\cite{robloxlbf}.

In such critical infrastructures, adversaries may attempt to craft false positives, causing the Bloom filter to deviate from its expected behavior~\cite{evilchoicesbloom}, thereby disrupting system operation. For example, as discussed in~\cite{moni1}, in a Bloom filter-based spam email whitelist (that stores addresses of known senders), crafting false positives allows spam emails to pass as benign emails. Similarly, in web caches, a Bloom filter can be used to represent the web pages in the cache; a false positive causes an unsuccessful cache access, which degrades performance and may eventually cause a denial of service attack. Another case is related to databases~\cite{tirmazi2025lsm}, including Meta's RocksDB and Google's LevelDB; crafting false positives in the Bloom filters of RocksDB (respectively LevelDB) can degrade the performance of query lookups by up to 8x (respectively 2x).

These cases attest to the importance of studying provable security for Bloom filters against adaptive adversaries. This has been done for CBFs, including studying practical attacks~\cite{evilchoicesbloom}, and formalizing notions for adaptive security alongside showing new CBF constructions that realize these notions~\cite{moni1,naor2022bet}. However, to the best of our knowledge, no such notions/provably-secure constructions exist for LBFs. Indeed, as an LBF contains a backup CBF, it inherits all the adversarial vulnerabilities of CBFs. However, securing the backup CBF does not imply that the LBF employing this CBF is secure, as LBFs have different designs. That is, having a learning model allows adversaries to craft false positives in new ways~\cite{reviriego1}. Examples include mutating an existing false positive or modifying certain features (relevant to the learning model) of a true negative input to convert it into a false positive. Nonetheless, the work~\cite{reviriego1} only demonstrated such attacks while leaving provable adaptive security of LBFs as an open problem.

\begin{figure}[t!]
  \centering
  \begin{tikzpicture}[
    node distance=0.5cm,
    box/.style={draw, minimum width=2cm, minimum height=0.5cm, align=center},
    arrow/.style={-{Stealth[length=3mm]}, thick},
    label/.style={font=\small}
]

\node[box] (LM) {Learning Model};
\node[box, right=2cm of LM] (CB) {Backup CBF};

\draw[arrow] ([xshift=-1.5cm]LM.west) -- node[label, above] {Input} (LM.west);

\draw[arrow] (LM.east) -- node[label, above] {Negatives} (CB.west);

\draw[arrow] (CB.east) -- node[label, above] {Negatives} ([xshift=1.5cm]CB.east);

\draw[arrow] (LM.south) -- +(0, -0.8cm) node[pos=0.3, label, right] {Positives};
\draw[arrow] (CB.south) -- +(0, -0.8cm) node[pos=0.3, label, right] {Positives};

\end{tikzpicture}
  \caption{Conventional LBF architecture---a backup CBF only checks values that are identified as (highly probably) negative by the learning model.}\label{fig:learnedbloom}
\end{figure}

\subsection{Contributions}
\label{subsec:contributions}
In this paper, we close this gap and initiate a formal study on the security of LBFs against adaptive adversaries. In particular, we make the following contributions.\medskip

\noindent\textbf{New notions for adaptive security.} We define several new security notions for LBFs that capture a spectrum of adversarial capabilities. These include \emph{fully-adaptive} adversaries who choose all queries, \emph{partially-adaptive} adversaries who choose a fraction of queries, \emph{bet-or-pass} adversaries who can selectively choose to output a guess in the security game, and \emph{always-bet} adversaries who always output a guess (the latter two are inspired by similar models defined for CBFs~\cite{naor2022bet}). These models provide fine-grained understanding of the attack surface against LBFs and several security-efficiency tradeoffs. Moreover, we explore the relationships between these various notions, proving which security notions imply which. We believe that our notions contribute to establishing a foundation for analyzing learning-augmented data structures in adversarial settings.\medskip

\noindent\textbf{Constructions.} We propose two constructions for LBFs: \bloomname{} and \betpassbloomname{}. Our constructions extend a prior LBF variant called Partitioned LBFs~\cite{plbf} by introducing cryptographic hardness through pseudorandom permutations (PRPs) and pseudorandom functions (PRFs). Assuming one-way functions exist, we show that \bloomname{} achieves security under the fully-adaptive and partially-adaptive adversarial models. We also show (assuming one-way functions) that \betpassbloomname{} satisfies the bet-or-pass security notion (the extended notion for LBFs). Bet-or-pass is one of the strongest adversaries considered by prior work~\cite{naor2022bet,lotan2025adversarially} for CBFs. Our constructions require only $\bigO{n\log{\frac{1}{\varepsilon}} + \lambda}$ additional bits of memory, where $n$ is the cardinality of the set represented by the filter, $\varepsilon$ is the desired FPR, and $\lambda$ is the security parameter.\medskip

\noindent\textbf{Performance evaluation}. We analyze FPR of our constructions in both adversarial and non-adversarial settings. Our results show that LBFs can provide lower FPR than CBFs in realistic workloads using the same memory budget. They also show that our LBF constructions achieve better FPR-memory tradeoffs than prior CBF constructions while maintaining strong security guarantees. These results demonstrate that adversarial resilience can be supported for LBFs with minimal performance overhead.

\subsection{Related Work}
\label{subsec:related}
\noindent \textbf{Security of Classical Bloom Filters.} Gerbet et al.~\cite{evilchoicesbloom} demonstrate practical attacks on CBFs in the context of web crawlers and spam email filtering, and propose combining universal hash functions with message authentication codes to mitigate a subset of these attacks. Naor et al.~\cite{moni1} define an adversarial model for CBFs and show that for computationally bounded adversaries, non-trivial adversary 
resilient CBFs exist if and only if one-way functions exist, and that for computationally unbounded adversaries there exists a CBF
that is secure when the adversary makes $t$ queries while using only $\mathcal{O}(n \log \frac{1}{\varepsilon} + t)$ bits of memory. 

Clayton et al.~\cite{claytonetal} and Filic et al.~\cite{filic1} present stronger adversarial models than~\cite{moni1}, giving the adversary the capability of performing insertions and accessing the internal CBF state, and they show secure CBF constructions realizing these models. Naor and Oved~\cite{naor2022bet} introduce a comprehensive study of CBF security and define several robustness notions in a generalized adversarial model. Concretely,~\cite{naor2022bet} calls the notion of~\cite{moni1} as \textit{Always-Bet (AB)} test since the adversary must output a guess. Then, it introduces a new, strictly stronger, security notion called the \textit{Bet-Or-Pass (BP)} test, giving the adversary the option to pass without outputting a guess. Lotan and Naor~\cite{lotan2025adversarially} provide further results related to relationships between the security notions of~\cite{naor2022bet}.\medskip

\noindent\textbf{Security of Learned Bloom Filters.} LBFs were first introduced by Kraska et al.~\cite{learnedindexstructures} who showed that LBFs can offer better FPR vs. memory tradeoffs than CBFs. Mitzenmacher~\cite{learnedbloomsandwiching} provides the first rigorous mathematical model for LBFs, focusing on analyzing performance in terms of memory size and FPR, and introduces an LBF variant called Sandwiched LBF. Vaidia et al.~\cite{plbf} introduce another LBF variant, called partitioned LBFs (PLBFs). PLBFs use a learning model to partition the set $S$ into $p \in \NN$ partitions and use a separate backup CBF for each of the $p$ partitions. PLBFs reduce the FPR for specific partitions compared to original LBFs that employ only one backup CBF. 

All these works only examine performance, but not security. We are only aware of one prior work that studied LBFs in adversarial settings: Reviriego et al.~\cite{reviriego1} demonstrate practical attacks on LBFs, which, as mentioned earlier, rely on exploiting the learning model to craft false positives in new ways. In response, they propose two potential mitigations: the first relies on switching back to a CBF once an attack is detected, while the second adds a second backup CBF. However, Reviriego et al.~\cite{reviriego1} do not provide provable security guarantees for the suggested mitigations, or even empirical evaluations of their effectiveness, and leave provable security of LBFs as an open problem.

\subsection{Future Directions}

Our work offers important foundational steps towards understanding provable security of LBFs. Here, we list some natural directions for future work.\medskip

\noindent\textbf{Unsteady setting.} We provide secure constructions for LBFs in which the query algorithm \emph{does not} modify the internal representation of the LBF. Naor and Yogev~\cite{moni1} call this the \textit{steady} setting, and show secure CBF constructions for the \textit{unsteady} setting where the query algorithm changes the internal filter representation after each query. An interesting direction is to build LBF constructions and prove their adaptive security in the unsteady setting.\medskip

\noindent\textbf{Dynamic Bloom filters.} As opposed to static Bloom filters, which do not modify the input set $S$ after construction, dynamic Bloom filters allow insertions of new elements $S^{\prime} = \{ x \} \cup S$ after construction. Note that to maintain their one-sided error guarantees, even dynamic Bloom filters do not allow deletions. Similar to prior work~\cite{moni1,naor2022bet,lotan2025adversarially} that consider static CBFs, our work focuses on static LBFs. Clayton et al.~\cite{claytonetal} and Filic et al.~\cite{filic1} show provably-secure security constructions for dynamic CBFs. As such, another interesting direction to explore is formulating adaptive security and building secure constructions for dynamic LBFs under an adversarial model that admits insertions.\medskip
    
\noindent\textbf{Computationally unbounded adversaries.} Our constructions are secure against polynomial-time adversaries. Another interesting avenue for future research is showing LBF constructions that realize our security notions while assuming computationally unbounded adversaries.

%% file: tex/prelim.tex
\section{Preliminaries}
\label{sec:prelim}

In this section, we review CBFs and their adversarial models from~\cite{moni1,naor2022bet}, which form the basis for the adversarial models we define for LBFs.\medskip

\noindent\textbf{Notation.} For a set $S$, $|S|$ denotes the size of $S$, and $x \sample S$ denotes that $x$ is sampled uniformly at random from $S$. For $n \in \mathbb{N}$, $[n]$ denotes $\{ 1, \cdots, n\}$. Bloom filters can store elements from a finite domain $\domain$, so we have $S \subseteq \domain$. Lastly, $\lambda$ denotes the security parameter, $\negl$ denotes a function negligible in $\lambda$, and $\ppt$ denotes probabilistic polynomial time.\medskip

\noindent\textbf{CBF Modeling.} We adopt the CBF model from~\cite{moni1} while considering the \textit{steady} setting in which the query algorithm does not change the filter representation.

\begin{definition}\label{def:naoryogevbloom}
A Bloom filter $\mathbf{B} = (\mathbf{B}_{1}, \mathbf{B}_{2})$ is a pair of polynomial time algorithms. $\mathbf{B}_{1}$ is a randomized construction algorithm that takes as input a set $S \subseteq \domain$ of size $n$, and outputs a representation $M$. $\mathbf{B}_{2}$ is a deterministic query algorithm that takes as input an element $x \in \domain$ and a representation $M$, and outputs 1 indicating that $x \in S$, and 0 otherwise. We say that $\mathbf{B}$ is an $(n, \varepsilon)$-Bloom filter if for all sets $S \subseteq \domain$ of cardinality $n$, the following hold:

\begin{enumerate}
    \item Completeness: $\forall{x} \in S : \Pr[\mathbf{B}_{2}(\mathbf{B}_{1}(S), x) = 1] = 1$
    \item Soundness: $\forall x \notin S : \Pr[\mathbf{B}_{2}(\mathbf{B}_{1}(S), x) = 1] \leq \varepsilon$ 
\end{enumerate}

\noindent where the probabilities are taken over the random coins of $\mathbf{B}_{1}$.

\end{definition}

A standard CBF is implemented as a string $str$ of length $n_{b} \in \NN$ bits indexed over $[n_{b}]$, along with $n_{h} \in \NN$ hash functions $h_i$, for $i \in [n_h]$, from a universal hash family $\mathcal{H}$ used for element mapping. That is, each $h_{i}$ maps an element from $\domain$ to an index value within $[n_{b}]$, so $h_{i}: \domain \mapsto [n_{b}]$. Then, for each element $x \in S$, for $i \in [n_h]$, the bit at index $h_{i}(x)$ of $str$ is set to $1$ (if it is already set to 1, it stays 1). For a queried element $x$, a CBF returns 1 if all bits within $str$ at indices corresponding to $h_{i}(x)$ are set to $1$, otherwise, the CBF returns 0.\medskip

\noindent\textbf{Always-Bet Security.}
Naor et al.\cite{moni1,naor2022bet} introduced the Always-Bet security game for CBFs, denoted as \texttt{ABGame}, in which an adversary who has oracle access to the CBF aims to find a false positive. It is denoted as always-bet since the adversary is required to output a guess $x^*$ at the end of the game that will be tested whether it is a false positive (i.e., always bets that $x^*$ is a false positive). In detail, for an adversary $\adv = (\adv_{1}, \adv_{2})$, ${\adv}_{1}$ chooses a set $S \subseteq \domain$ for which $\mathbf{B}_1$ will compute a representation $M$, and ${\adv}_{2}$ gets $S$ as input and attempts to find a false positive $x^{*}$ that was not queried before, given only oracle access $\mathbf{B}_{2}(M, \cdot)$.\medskip

\noindent \texttt{ABGame($\adv, t, \lambda$)}:
    \begin{enumerate}
    \vspace{-3pt}
        \item $S \leftarrow \adv_{{1}}(1^{\lambda})$ such that $S \subseteq \domain$ and $|S| = n$.
        \item $M \leftarrow \mathbf{B}_{1}(S)$.
        \item $x^{*} \leftarrow \adv_{2}^{\mathbf{B}_2(M, \cdot)}(1^{\lambda}, S)$ where $\adv_{2}$ can make at most $t$ queries to $\mathbf{B}_{2}(M, \cdot)$---in this and other games, we denote the adversarial queries as $\{x_1, \dots, x_t$\}.
        \item If $x^{*} \notin S \cup \{ x_{1}, \cdots, x_{t} \}$ and $\mathbf{B}_{2}(M, x^*) = 1$, return $1$. Otherwise, return $0$. 
    \end{enumerate}

We use the following security notion for all security games discussed in this paper (except for \texttt{BPGame}, which we discuss next). 

\begin{definition}\label{def:securecbf}
    A Bloom filter $\mathbf{B}$ is $(n, t, \varepsilon)$-secure under a security game $\texttt{Game}$ if for for all large enough $\lambda \in \mathbb{N}$, all $\ppt$ adversaries $\adv$, and all sets $S \subseteq \domain$ of cardinality $n$, it holds that $
    \text{Pr}[\texttt{Game}(\adv, t, \lambda) =1] \leq \varepsilon$,  where the probability is taken over the random coins of $\mathbf{B}$ and $\adv$.
\end{definition}

\noindent\textbf{Bet-or-Pass Security.}
Naor and Oved~\cite{naor2022bet} introduce a security game stronger than \texttt{ABGame} called Bet-or-Pass, or \texttt{BPGame}. Here, $\adv$ can either output $x^{*}$ or \textit{pass}, so it is not always betting on the output of the game to be a false positive. \texttt{BPGame} also defines $\adv$'s \textit{profit} $C_{\adv}$, rewarding $\adv$ if $x^{*}$ is a false positive, penalizing $\adv$ if $x^{*}$ is not a false positive, and leaving $C_{A}$ unchanged if $\adv$ chooses to pass.\medskip 

\noindent \texttt{BPGame($\adv, t, \lambda$)}:
    \begin{enumerate}
    \vspace{-3pt}
        \item $S \leftarrow \adv_{{1}}(1^{\lambda})$ such that $S \subseteq \domain$ and $|S| = n$.
        \item $M \leftarrow \mathbf{B}_{1}(S)$.
        \item $(b, x^{*}) \leftarrow \adv_{2}^{\mathbf{B}_2(M, \cdot)}(1^{\lambda}, S)$ where $\adv_{2}$ can make at most $t$ queries, $\{x_1, \dots, x_t\}$, to $\mathbf{B}_{2}(M, \cdot)$.
        \item  Return $\adv$'s profit, $C_{\adv}$, which is defined as \begin{equation*} C_{\adv} = \begin{cases} 0, & \text{if } b = 0\\
        \varepsilon^{-1}, & \begin{aligned}[t] &\text{if } b = 1 \text{ and } x^{*} \notin S \cup \{ x_{1}, \ldots, x_{t} \} \\ &\text{ and } \mathbf{B}_{2}(M, x^*) = 1 \end{aligned}\\
        -(1 - \varepsilon)^{-1}, & \text{otherwise} 
        \end{cases} \end{equation*}
    \end{enumerate}
The profit formulation is set in this way to ensure that an adversary that makes a random guess (which will be a false positive with probability $\varepsilon$) has an expected profit of zero~\cite{naor2022bet}. $\adv$ breaks the security of the CBF if its profit is noticeably larger than zero. Thus, the security guarantee for \texttt{BPGame} is defined as an upper bound on the expectation of the adversary's profit.


\begin{definition}\label{def:securebetorpass}
    A Bloom filter $\mathbf{B}$ is $(n, t, \varepsilon)$-secure under \texttt{BPGame} if for all large enough $\lambda \in \mathbb{N}$,  all $\ppt$ adversaries $\adv$, and all sets $S \subseteq \domain$ of cardinality $n$, there exists a negligible function $\negl[\cdot]$ such that $\expect{C_{\adv}} \leq \negl$, where the expectation is taken over the random coins of $\mathbf{B}$ and $\adv$.
\end{definition}

%% file: tex/model.tex
\section{Definitions and Adversarial Models for LBFs}
\label{sec:lbf-adversary-model}

In this section, we first define a model for LBFs, followed by three LBF adaptive security notions: full adaptivity (which when slightly modified captures always-bet security), partial adaptivity, and learned bet-or-pass.  

\subsection{LBF Definition}
We present our definitions for LBFs, which are based on the model of~\cite{learnedbloomsandwiching} with additional formalism and adaptations to make it convenient to compare with the models of CBFs introduced earlier. An LBF uses a learning model trained over the dataset the LBF represents, such that the model determines a function $\mathscr{L}$ that models this set. In particular, on input $x$,  $\mathscr{L}$ outputs the probability that $x \in S$. In what follows, we define the notion of a training dataset and a learning model in the context of LBFs (where $y_i$ is a label stating whether $x_i \in S$).

\begin{definition}\label{def:trainingdataset}
Let $S \subseteq \domain$ be any set represented by a Bloom filter. For any two sets $P \subseteq S$ and $N \subseteq \domain \setminus S$, the training dataset is the set $\mathscr{T} = \{ (x_{i}, y_{i} = 1) \mid x_{i} \in P \} \cup \{ (x_{i}, y_{i} = 0) \mid x_{i} \in N\}$.
\end{definition}

\begin{definition}\label{def:learning_model} 
  For an $\mathscr{L}: \domain \mapsto [0, 1]$ and threshold $\tau$, we say $\mathscr{L}$ is an $(S, \tau, \varepsilon_{p}, \varepsilon_{n})$-learning model if for any set $S \subseteq \domain$ the following hold:
  \begin{enumerate}
      \item P-Soundness: $\forall x \notin S : \Pr[\mathscr{L}(x) \geq \tau] \leq \varepsilon_{p}$
      \item N-Soundness: $\forall x \in S : \Pr[\mathscr{L}(x) < \tau] \leq \varepsilon_{n}$
  \end{enumerate}
  \noindent where the probability is taken over the random coins of $\mathscr{L}$.
\end{definition}

Now, we define an LBF capturing both the classic and learning components. We consider the \textit{steady} setting in which the query algorithm $\mathbf{B}_{2}$ does not change the learned representation of the set $S$ (including both the set representation held by the backup CBF, i.e., $M$, and $(\mathscr{L}, \tau)$). 

\begin{definition}\label{def:learnedbloomfilter} 
An LBF $\mathbf{B} = (\mathbf{B}_{1}, \mathbf{B}_{2}, \mathbf{B}_{3}, \mathbf{B}_{4})$ is a tuple of four polynomial-time algorithms: $\mathbf{B}_1$ is as before, $\mathbf{B}_2$ is a query algorithm, $\mathbf{B}_{3}$ is a randomized algorithm that takes a set $S \subseteq \domain$ as input and outputs a training dataset $\mathscr{T}$, and $\mathbf{B}_{4}$ is a randomized algorithm that takes the training dataset $\mathscr{T}$ as input and returns a learning model $\mathscr{L}$ and a threshold $\tau \in [0, 1]$. The internal representation of an LBF contains two components: the classical component $M$ and the learned component ($\mathscr{L}, \tau)$. $\mathbf{B}_{2}$ takes as inputs an element $x \in \domain$, $M$, and ($\mathscr{L}, \tau$), and outputs 1, indicating that $x \in S$, and 0 otherwise. We say that $\mathbf{B}$ is an $(n, \tau, \varepsilon, \varepsilon_{p}, \varepsilon_{n})$-LBF if for all sets $S \subseteq \domain$ of cardinality $n$, it holds that:

\begin{enumerate}
    \item Completeness: $\forall{x} \in S : \Pr[\mathbf{B}_{2}(\mathbf{B}_{1}(S), \mathbf{B}_{4}(\mathbf{B}_{3}(S)), x) = 1] = 1$.
    \item Filter soundness: $\forall x \notin S : \Pr[\mathbf{B}_{2}(\mathbf{B}_{1}(S), \mathbf{B}_{4}(\mathbf{B}_{3}(S)), x) = 1] \leq \varepsilon$. 
    \item Learning model soundness: $\mathbf{B}_{4}(\mathbf{B}_{3}(S))$ is an $(S, \tau, \varepsilon_{p}, \varepsilon_{n})$-learning model.
\end{enumerate}
\noindent where the probabilities are over the random coins of $\mathbf{B}_{1}$, $\mathbf{B}_{3}$, and $\mathbf{B}_{4}$.

\end{definition}

Standard LBFs (Figure~\ref{fig:learnedbloom}) use a learning model as a pre-filter before a CBF. The CBF is called a backup CBF as it is only queried on inputs $x$ for which the learning model decides that $x$ is not an element of the stored set $S$, i.e., $\mathscr{L}(x) < \tau$.

\subsection{Full Adaptive Security}
\label{sec:fully_adaptive}

Full adaptivity means that $\adv$ chooses all the queries submitted to the query algorithm $\mathbf{B}_2$, i.e., the adversary controls the entire workload. For CBFs, the full adaptivity game, besides the oracle access to $\mathbf{B}_2(M, \cdot)$, gives oracle access to $\mathbf{B_1}$ enabling $\adv$ to obtain $M$ for any set $S$ of its choosing. Based on that, we define the following game capturing full adaptive security for CBFs.\medskip

\noindent\texttt{FAGame($\adv, t, \lambda$)}:
\vspace{-3pt}
    \begin{enumerate}
        \item $S \leftarrow \adv_{{1}}(1^{\lambda})$ such that $S \subseteq \domain$ and $|S| = n$.
        \item $M \leftarrow \mathbf{B}_{1}(S)$.
        \item $x^{*} \leftarrow \adv_{2}^{\oracle}(1^{\lambda}, S)$, where $\oracle = \{\mathbf{B}_1 (\cdot), \mathbf{B}_2 (M, \cdot)\}$. $\adv_{2}$ can make at most $t$ queries to $\mathbf{B}_2$, and any polynomial number of queries to $\mathbf{B}_1$.
        \item If $x^{*} \notin S \cup \{ x_{1}, \cdots, x_{t} \}$ and $\mathbf{B}_2 (M, x^*) = 1$, return $1$, else, return $0$. 
    \end{enumerate}

For LBFs, we define a similar security game, denoted \texttt{LFAGame}. The difference is that $\adv$ now has oracle access to the additional algorithms in the LBF structure. We note that our results hold even if we let $\adv$ choose $\mathscr{T}$ directly (rather than having the challenger invoke $\mathbf{B}_3$ over the set $S$ chosen by $\adv$), as long as the challenger validates that this $\mathscr{T}$ satisfies Definition~\ref{def:trainingdataset}.\medskip

\noindent\texttt{LFAGame}($\adv, t, \lambda$):
\vspace{-3pt}
    \begin{enumerate}
        \item $S \leftarrow \adv_{{1}}(1^{\lambda})$ such that $S \subseteq \domain$ and $|S| = n$.
        \item $M \leftarrow \mathbf{B}_{1}(S)$, $\mathscr{T} \leftarrow \mathbf{B}_{3}(S)$, and $(\mathscr{L}, \tau) \leftarrow \mathbf{B}_{4}(\mathscr{T})$.
        \item $x^{*} \leftarrow \adv_{2}^{\oracle}(1^{\lambda}, S)$, where $\oracle = \{\mathbf{B}_1 (\cdot), \mathbf{B}_2 (M, \mathscr{L}, \tau, \cdot), \mathbf{B}_3(\cdot), \mathbf{B}_4(\cdot)\}$. $\adv_{2}$ can make at most $t$ queries to $\mathbf{B}_2$, and any polynomial number of queries to each of $\mathbf{B}_1$, $\mathbf{B}_3$, and $\mathbf{B}_4$.
        \item If $x^{*} \notin S \cup \{ x_{1}, \cdots, x_{t} \}$ and $\mathbf{B}_2 (M, \mathscr{L},\tau, x^*) = 1$, return 1, else, return 0. 
    \end{enumerate}

If we remove $\adv_{2}$'s oracle access to $\mathbf{B}_1(\cdot)$, $\mathbf{B}_3(\cdot)$, and $\mathbf{B}_4(\cdot)$, we obtain a notion for the \emph{always-bet} security game, which we refer to as \texttt{LABGame($\adv, t, \lambda$)}.

\subsection{Partial Adaptive Security}\label{sec:partially-adaptive-security}
For partial adaptivity, among the $t$ queries to $\mathbf{B}_2$, $\adv$ can choose $\alpha t$ of them, where $\alpha \in [0, 1]$. These $t$ queries may be part of a batch workload or a streaming workload under any streaming models described by~\cite{muthukrishnanetal}. Systems incorporating Bloom filters can operate under such a partial-adaptivity model in many real-world scenarios, including content caching (e.g., as in content delivery networks) and database systems. For example, LSM (log-structured merge) Tree stores, including Google's LevelDB~\cite{leveldbbloom} and Facebook's RocksDB~\cite{rocksdbbloom}, use Bloom filters to reduce read times~\cite{dayan2017monkey}. These stores can receive queries from both malicious and non-malicious users, which captures the fact that $\adv$ can observe the output of queries made by others while it can choose the rest. 

In the partial adaptivity game \texttt{PAGame}, as in \texttt{FAGame}, $\adv$'s goal is to produce a previously unseen false positive. However, this time, $\adv$ cannot choose all the queries; it can only choose a fraction $\alpha$ of them, while the remaining $(1 - \alpha)t$ queries are uniformly sampled at random from $\domain$. $\adv$, however, still observes the output of all queries. $\adv$ also has the freedom to choose the order in which adversarial queries are interleaved between non-adversarial queries. We show this notion first for CBFs, where $c$ indicates $\adv$'s choice of whether to evaluate the adversarial query if $c = 1$, or a non-adversarial one if $c = 0$.\medskip

\noindent\texttt{PAGame($\adv, \alpha, t, \lambda$)}:
\vspace{-3pt}
    \begin{enumerate}
        \item $S \leftarrow \adv_{{1}}(1^{\lambda})$ such that $S \subseteq \domain$ and $|S| = n$.
        \item $M \leftarrow \mathbf{B}_{1}(S)$, and set $i = \beta = 0$.
        \item $(c, x_i) \leftarrow \adv_{2}^{\mathbf{B}_1(\cdot)}(1^{\lambda}, S, i, \beta)$---$\adv_{2}$ can make any polynomial number of queries to $\mathbf{B}_1$.
        \item If $c = 1$ and $\beta < \alpha t$, give $\adv_{2}$ the output of $\mathbf{B}_2(M,x_i)$, and set $\beta = \beta + 1$.
        \item Otherwise, ${x_i} \sample \domain$ and give $\adv_{2}$ the output of $\mathbf{B}_2(M,x_i)$.
        \item Set $i = i + 1$. If $i < t$, go back to Step 3.
        \item $\adv_2$ outputs $x^{*}$. If $x^{*} \notin S \cup \{x_1, \dots, x_t\}$ and $\mathbf{B}_2(M,x^{*}) = 1$, return 1. Otherwise, return 0.
    \end{enumerate}

This security game can be modified to work for LBFs as follows.\medskip

\noindent\texttt{LPAGame($\adv, \alpha, t, \lambda$)}:
\vspace{-3pt}
    \begin{enumerate}
        \item $S \leftarrow \adv_{{1}}(1^{\lambda})$ such that $S \subseteq \domain$ and $|S| = n$.
        \item $M \leftarrow \mathbf{B}_{1}(S)$, $\mathscr{T} \leftarrow \mathbf{B}_3(S)$, $(\mathscr{L}, \tau) \leftarrow \mathbf{B}_4(\mathscr{T})$, and set $i = \beta = 0$.
        \item $(c, x_i) \leftarrow \adv_{2}^{\oracle}(1^{\lambda}, S, i, \beta)$, where $\oracle = \{\mathbf{B}_1 (\cdot), \mathbf{B}_3(\cdot), \mathbf{B}_4(\cdot)\}$ and $\adv_{2}$ can make any polynomial number of queries to $\oracle$.
        \item If $c = 1$ and $\beta < \alpha t$, give $\adv_{2}$ the output of $\mathbf{B}_2(M,\mathscr{L}, \tau,x_i)$, and set $\beta = \beta + 1$.
        \item Otherwise, ${x_i} \sample \domain$ and give $\adv_{2}$ the output of $\mathbf{B}_2(M,\mathscr{L}, \tau,x_i)$.
        \item Set $i = i + 1$. If $i < t$, go back to Step 3.
        \item $\adv_2$ outputs $x^{*}$. If $x^{*} \notin S \cup \{x_1, \dots, x_t\}$ and $\mathbf{B}_2(M,\mathscr{L},\tau,x^{*}) = 1$, return 1. Otherwise, return 0.
    \end{enumerate}

Definition~\ref{def:securecbf} still applies to these games with one change; now we say a Bloom filter is $(n,\alpha,t,\varepsilon)$-secure to account for the additional parameter $\alpha$.

\subsection{Learned Bet-or-Pass Security}

This section extends the CBF \texttt{BPGame} from Naor and Oved~\cite{naor2022bet} to LBFs, denoted as \texttt{LBPGame} (the same expectation notion from Definition~\ref{def:securebetorpass} applies here as well).\medskip

\noindent\texttt{LBPGame($\adv, t, \lambda$)}:
\vspace{-3pt}
    \begin{enumerate}
        \item $S \leftarrow \adv_{{1}}(1^{\lambda})$ such that $S \subseteq \domain$ and $|S| = n$.
        \item $M \leftarrow \mathbf{B}_{1}(S)$, $\mathscr{T} \leftarrow \mathbf{B}_3(S)$, $(\mathscr{L}, \tau) \leftarrow \mathbf{B}_4(\mathscr{T})$.
        \item $(b, x^{*}) \leftarrow \adv_{2}^{\oracle}(1^{\lambda}, S)$ where $\oracle = \{\mathbf{B}_1 (\cdot), \mathbf{B}_2 (M, \mathscr{L}, \tau, \cdot), \mathbf{B}_3(\cdot), \mathbf{B}_4(\cdot)\}$. $\adv_{2}$ can make at most $t$ queries to $\mathbf{B}_2$, and any polynomial number of queries to each of $\mathbf{B}_1$, $\mathbf{B}_3$, and $\mathbf{B}_4$.
        \item  Return $\adv$'s profit, $C_{\adv}$, which is defined as \begin{equation*} C_{\adv} = \begin{cases} 0, & \text{if } b = 0\\
        \varepsilon^{-1}, & \begin{aligned}[t] &\text{if } b = 1 \text{ and } x^{*} \notin S \cup \{ x_{1}, \ldots, x_{t} \} \\ &\text{ and } \mathbf{B}_{2}(M, \mathscr{L},\tau, x^*) = 1 \end{aligned}\\
        -(1 - \varepsilon)^{-1}, & \text{otherwise} 
        \end{cases} \end{equation*}
    \end{enumerate}

\subsection{Relationships between Security Notions}
\label{sec:relation_security_games}
We investigate relationships between the security games we defined so far. For clarity, we focus on LBF versions of all security games. Naor et al.~\cite{naor2022bet,lotan2025adversarially} investigate similar relationships for the CBF security notions they defined, including \texttt{ABGame} and \texttt{BPGame}. 

We first explore connections between \texttt{LFAGame} and \texttt{LPAGame}. Notice that when $\alpha = 1$, \texttt{LPAGame} is equivalent to \texttt{LFAGame}. \texttt{LFAGame}, therefore, is the special case of \texttt{LPAGame} where all queries are adversarial. The converse relationship is false, meaning that \texttt{LFAGame} is stronger than \texttt{LPAGame}.

\begin{theorem}\label{thm:lfaimplieslpa}
For $\varepsilon \in (0,1), \alpha \in [0, 1]$ and $n,t \in \NN$, we have: \[(n, t, \varepsilon)\text{-security in \texttt{LFAGame}} \implies (n, \alpha, t, \varepsilon)\text{-security in \texttt{LPAGame}}\]
\[(n, \alpha, t, \varepsilon)\text{-security in } \texttt{LPAGame} \notimplies (n, t, \varepsilon)\text{-security in } \texttt{LFAGame}\]
\end{theorem}

\begin{proof}
For the first relationship, fix a $\ppt$ adversary $\adv = (\adv_{1}, \adv_{2})$ and $t, \lambda \in \NN$. Pick any $\alpha, \alpha^{\prime} \in [0, 1]$ such that $\alpha^{\prime} \geq \alpha$. Construct another adversary $\adv^{\prime} = (\adv^{\prime}_{1}, \adv^{\prime}_{2})$ as follows. Let $\adv^{\prime}_{1}(1^{\lambda}, \alpha^{\prime}, t) = \adv_{1}(1^{\lambda}, \alpha, t)$ and $\adv^{\prime}_{2}(1^{\lambda}, S, \alpha^{\prime}, t, \beta, i) = \adv_{2}(1^{\lambda}, S, \alpha, t, \beta, i)$. This does not break the rules of $\mathtt{LPAGame}$ as still $\alpha' \in [0,1]$. Since $\alpha^{\prime} \geq \alpha$, the winning probability of $\adv^{\prime}$ is at least as high as $\adv$'s winning probability. When $\alpha^{\prime} = 1$, we have $\mathtt{LPAGame}(\adv', 1, t, \lambda) = \mathtt{LFAGame}(\adv', t, \lambda)$, thus the result follows.


For the second relationship, let $\mathbf{B}$ be an $(n, \alpha, t, \varepsilon)$-secure LBF construction in $\mathtt{LPAGame}$. Consider $\mathtt{LFAGame}$ with the same parameters (excluding $\alpha$) as $\mathtt{LPAGame}$. Also, let $\mathbf{B}^{\prime}$ be an alternative LBF construction such that the construction algorithm  $\mathbf{B}^{\prime}$ is identical to $\mathbf{B}$, while the query algorithm $\mathbf{B}^{\prime}_{2}$ differs from $\mathbf{B}_{2}$ as follows. Initially, $\mathbf{B}^{\prime}_{2}(M, \mathscr{L}, \tau, \cdot) = \mathbf{B}_{2}(M, \mathscr{L}, \tau, \cdot)$. However, $\mathbf{B}^{\prime}_{2}$ tracks the number of identical queries it receives. On receiving $t$ identical queries, $\mathbf{B}^{\prime}_{2}$ switches to always output $1$ for all future queries. An adversary can win \texttt{LFAGame} with non-negligible probability by having all its $t$ queries over the same value $x$. At this point, $\mathbf{B}^{\prime}_{2}$ will output 1 for any guess $x^{*} \neq x$ that the adversary outputs and wins the game. Therefore, $\mathbf{B}^{\prime}$ is not $(n, \alpha, t, \varepsilon)$-secure in \texttt{LFAGame}. Now, consider an adversary $\adv$ in \texttt{LPAGame}. To win with non-negligible probability, $\adv$ must trigger $\mathbf{B}^{\prime}_{2}$ to always output $1$. $\adv$ can have at most $t - 1$ adversarial queries, and can choose all of them to be over the same value $x$. However, to trigger $\mathbf{B}^{\prime}_{2}$ to always output 1, the final query, which is chosen uniformly randomly by the challenger, must also be $x$. Assuming the domain $\domain$ is large, the probability of this happens is negligible. Therefore, $\mathbf{B}^{\prime}$ remains $(n, \alpha, t, \varepsilon)$-secure in \texttt{LPAGame}.\qed


\end{proof}
We now compare $\texttt{LFAGame}$ and $\texttt{LPAGame}$ to $\texttt{LABGame}$ and $\texttt{LBPGame}$. 

\begin{theorem}\label{thm:lfa_implies_lab}
For $\varepsilon \in (0,1), \alpha \in [0, 1]$ and $n,t \in \NN$, we have:
\[(n, t, \varepsilon)\text{-security in } \texttt{LFAGame} \implies (n, t, \varepsilon)\text{-security in } \texttt{LABGame}\]
\[(n, \alpha, t, \varepsilon)\text{-security in } \texttt{LPAGame} \implies (n, \alpha t, \varepsilon)\text{-security in } \texttt{LABGame}\]
\end{theorem}
\begin{proof}
Fix a $\ppt$ adversary $\adv$, and $t, \lambda \in \NN$. $\mathtt{LABGame}(\adv, t, \lambda)$ is identical to $\mathtt{LFAGame}(\adv, t, \lambda)$ aside from the fact that $\adv$ does not have oracle access to $\mathbf{B}_{1}(\cdot)$ in $\mathtt{LABGame}$. Hence, $\adv$'s winning probability in $\mathtt{LABGame}(\adv, t, \lambda)$ cannot be greater than $\adv$'s winning probability in $\mathtt{LFAGame}(\adv, t, \lambda)$. Following a similar argument, we get an implication between \texttt{LPAGame} and \texttt{LABGame}.\qed
\end{proof}

To disprove the converse, we introduce an $(n, t, \varepsilon)$-secure construction under \texttt{LABGame} and then modify it such that it remains secure under \texttt{LABGame} but not under \texttt{LFAGame}. We first recall the following theorem for CBFs by Naor and Yogev~\cite{moni1}.

\begin{theorem}[Naor-Yogev Theorem]\label{thm:moninaortheorem}
Let $\textbf{B}$ be an $(n, \varepsilon)$-Bloom filter using $m$ memory bits. If pseudorandom permutations (PRPs) exist, then for security parameter $\lambda$ there exists a negligible function $\negl[\cdot]$ and an $(n,  \varepsilon + \negl)$-strongly resilient Bloom filter in \texttt{ABGame} that uses $m^{\prime} = m + \lambda$ bits of memory.\footnote{Strongly resilient means being $(n, t, \varepsilon)$-secure under \texttt{BPGame} for any $t \in \bigO{\poly[n, \lambda]}$. For more details, see Definition 2.4 in~\cite{naor2022bet}.}
\end{theorem}

This theorem is proved for the following construction, which we denote as Naor-Yogev CBF, or \securecbf{}. Run the initialization algorithm of a CBF with the set $S^{\prime} = {\{\prp_{\sk}(x): x \in S\}}$ instead of $S$, where $\prp_{\sk}$ is a keyed PRP. Similarly, for an element $x \in \domain$, query the filter over $\prp_{\sk}(x)$ instead of $x$. This new CBF construction uses $m + \lambda$ bits of memory and is $(n,  \varepsilon + \negl)$-secure for any $t \in \bigO{\poly[n, \lambda]}$ under \texttt{ABGame}. We modify this construction in the proof for the theorem below.

\begin{theorem}\label{thm:labnotimplylpalfa}
For $\varepsilon \in (0,1), \alpha \in [0, 1]$, $n,t \in \NN$, and $\delta \in (0, 1)$, we have:
\[(n, t, \varepsilon)\text{-security in \texttt{LABGame}} \notimplies (n, t, \delta)\text{-security in \texttt{LFAGame}}\]
\[(n, \alpha t, \varepsilon)\text{-security in \texttt{LABGame}} \notimplies (n, \alpha, t, \delta)\text{-security in \texttt{LPAGame}}\]
\end{theorem}

\begin{proof}
    Let $\mathbf{B}$ be an \securecbf{} that is $(n, t, \varepsilon)$-secure under $\texttt{ABGame}$. Consider a construction $\mathbf{B}^{\prime}$ that is $(n, t, \varepsilon)$-secure under $\texttt{LABGame}$. $\mathbf{B}^{\prime}$ replaces $\mathbf{B}$ with standard (or conventional) LBF that uses $\mathbf{B}$ as its backup CBF and has a trivial learning model that responds negative (i.e., $x \notin S$) to all queries $x \in \domain$. Although contrived, $\mathbf{B}^{\prime}$ is a correct LBF by Definition~\ref{def:learnedbloomfilter}. Consider a $\ppt$ adversary $\adv$. Since the learned representation of $\mathbf{B}^{\prime}$ contains no information on the input set $S$ and routes all queries to its backup CBF, which is a \securecbf{}, $\adv$ in $\texttt{LABGame}$ has no advantage over $\adv$ in $\texttt{ABGame}$. Hence, $\mathbf{B}^{\prime}$ is also $(n, t, \varepsilon)$-secure in $\texttt{LABGame}$. 
    
    Now consider a second construction $\mathbf{B}^{\prime\prime}$ that is identical to $\mathbf{B}^{\prime}$ but with one modification: instead of storing the internal representation $M$ like $\mathbf{B}^{\prime}$, $\mathbf{B}^{\prime\prime}$ stores the internal representation $M^{\prime} = (M, \sk)$ where $\sk$ is the secret key of $\prp$. $\mathbf{B}^{\prime\prime}$ is still $(n, t, \delta = \varepsilon + \negl)$-secure under \texttt{LABGame} as $\adv$ does not have access to the internal representation of the Bloom filter. However, $\mathbf{B}^{\prime\prime}$ is not $(n, t, \delta = \varepsilon + \negl)$-secure under \texttt{LFAGame} or \texttt{LPAGame} where $\adv$ has oracle access to $M^{\prime}$ and can obtain the secret key $\sk$.\qed
\end{proof}



Before we show our next result, we recall a theorem by Naor and Oved~\cite{naor2022bet}, in which they proved that security in \texttt{ABGame} does not imply security in \texttt{BPGame}. Their proof uses a counterexample construction that usually behaves like an $(n, t, \varepsilon)$-secure CBF in \texttt{ABGame} but has a small probability of reaching an always-one state, i.e., the query algorithm always outputs 1 for any query. 

\begin{theorem}[Naor-Oved Theorem]\label{thm:naor-oved}
Let $0 < \varepsilon < 1$ and $n \in \mathbb{N}$, then for any $0 < \delta < 1$, assuming the existence of one-way functions, there exists a non-trivial Bloom filter $\mathbf{B}$ that is $(n,\varepsilon)$-strongly resilient under \texttt{ABGame} and is not $(n,\delta)$-strongly resilient under \texttt{BPGame}.
\end{theorem}

\begin{theorem}
For $\varepsilon \in (0,1)$, $n,t \in \NN$, and $\delta \in (0, 1)$, we have:
\[(n, t, \varepsilon)\text{-security in \texttt{LABGame}}\notimplies (n, t, \delta)\text{-security in \texttt{LBPGame}}\]
\[(n, t, \varepsilon)\text{-security in \texttt{LBPGame}}\implies (n, t, \varepsilon)\text{-security in \texttt{LFAGame} and \texttt{LABGame}}\]
\end{theorem}

\begin{proof}
For the first relationship, let $\mathbf{B}$ be the Naor-Oved CBF construction that is $(n, t, \varepsilon)$-secure under $\texttt{ABGame}$ but not under \texttt{BPGame}. We demonstrated in Theorem~\ref{thm:labnotimplylpalfa} how to create a construction $\mathbf{B}^{\prime}$ that is $(n, t, \varepsilon)$-secure under \texttt{LABGame} by having a trivial learning model that routes all queries to the backup CBF. This backup CBF is $\mathbf{B}$, which is secure under the \texttt{ABGame}, making the overall construction secure under the \texttt{LABGame}. By Theorem~\ref{thm:naor-oved}, we know that $\textbf{B}$ is not $(n, t, \delta)$-secure under \texttt{BPGame} for any $\delta \in (0, 1)$. Since the learning model in $\mathbf{B}^{\prime}$ is trivial and adds no adversarial resilience, it follows that $\mathbf{B}^{\prime}$ is not $(n, t, \delta)$-secure under \texttt{LBPGame} for any $\delta \in (0, 1)$. 

For the second relationship, Naor and Oved~\cite{naor2022bet} also prove that for CBFs $(n, t, \varepsilon)$-security under \texttt{BPGame} implies $(n, t, \varepsilon)$-security under \texttt{ABGame}. We use a proof technique similar to theirs. Let $\mathbf{B}$ be a construction that is $(n, t, \varepsilon)$-secure under $\texttt{LBPGame}$. Fix an \texttt{LBPGame} adversary $\adv$ that outputs a guess $x^{*}$. Let $\text{FP}$ denote the event that $x^{*}$ is a false positive. For construction $\mathbf{B}$, the expected profit of $\adv$ is:
\begin{align*}
\expect{C_A} &= \frac{1}{\varepsilon} \Pr[\text{FP}] - \frac{1}{1 - \varepsilon} \Pr[\lnot \text{FP}] = \frac{1}{\varepsilon} \Pr[\text{FP}] - \frac{1}{1 - \varepsilon} (1 - \Pr[\text{FP}])\\ 
&= \left( \frac{1}{\varepsilon} + \frac{1}{1 - \varepsilon} \right) \Pr[\text{FP}] - \frac{1}{1 - \varepsilon} = \frac{1}{\varepsilon (1 - \varepsilon)} \Pr[\text{FP}] - \frac{1}{1 - \varepsilon}
\end{align*}

Using $\expect{C_A} \leq \negl[\lambda]$, we obtain 
\[\frac{1}{\varepsilon (1 - \varepsilon)} \Pr[\text{FP}] - \frac{1}{1 - \varepsilon} \leq \negl[\lambda]\] 

Since $\varepsilon \in (0, 1)$, we have $\varepsilon (1 - \varepsilon) \in (0, 1)$, and thus 
\[\Pr[\text{FP}] \leq \varepsilon + \varepsilon (1 - \varepsilon) \negl[\lambda] \leq \varepsilon + \negl[\lambda]\]. 

We have shown that the probability of $x^{*}$ being a false positive is at most negligibly greater than $\varepsilon$. This is the condition needed for $\mathbf{B}$ to be $(n, t, \varepsilon)$-secure under \texttt{LFAGame} and \texttt{LABGame}.\qed
\end{proof}

Finally, we show that security in \texttt{LFAGame} does not imply security in \texttt{LBPGame}. As mentioned before, Naor and Oved proved Theorem~\ref{thm:naor-oved} using a CBF counterexample that has a small probability of reaching an always-one state. We extend this idea to LBFs and prove that security in \texttt{LFAGame} does not imply security in \texttt{LBPGame}. Our proof is simpler than Naor and Oved's proof for CBFs because we give the adversary an oracle access to the Bloom filter's internal representation.

\begin{theorem}
For $\varepsilon \in (0,1)$, $n,t \in \NN$, and $\delta \in (0, 1)$, we have:
\[(n, t, \varepsilon)\text{-security in \texttt{LFAGame}}\notimplies (n, t, \delta)\text{-security in \texttt{LBPGame}}\]
\end{theorem}
\begin{proof}
    Fix an $\varepsilon^{\prime} \in (0, \varepsilon)$ and let $\mathbf{B}$ be an LBF construction that is $(n, t, \varepsilon^{\prime})$-secure in \texttt{LFAGame} (we show the existence of such a construction in Theorem~\ref{thm:downtownbodegafilter}). Let $\mathbf{B}^{\prime}$ be a modified construction that behaves as follows. The construction algorithm $\mathbf{B}^{\prime}_{1}$ flips a bit $b$ with probability $p$ of being $1$. If $b = 1$, $\mathbf{B}^{\prime}$ always answers $1$ to any query. Otherwise, $\mathbf{B}^{\prime}$ behaves identically to $\mathbf{B}$. Let $M$ be the internal classical representation of $\mathbf{B}$. $\mathbf{B}^{\prime}$ stores $M^{\prime} = (M, b)$ as its internal classical representation. We first show that $\mathbf{B}^{\prime}$ is $(n, t, \varepsilon)$-secure in \texttt{LFAGame}. An adversary $\adv$ in \texttt{LFAGame} is always required to output a guess $x^{*}$. Let $\text{FP}$ denote the event that $x^{*}$ is a false positive. Note that $\Pr[\text{FP} \mid b = 1] = 1$ while $\Pr[\text{FP} \mid b = 0]$ is same as the probability of $\adv$ winning \texttt{LFAGame} with construction $\mathbf{B}$, i.e., $\Pr[\text{FP} \mid b = 0] \leq \varepsilon^{\prime} + \negl[\lambda]$. Based on that, and by choosing $p = \frac{\varepsilon - \varepsilon^{\prime}}{1 - \varepsilon^{\prime}}$, we have: 
    \begin{align*}
    \Pr[\text{FP}] &= \Pr[\text{FP} \mid b = 1] \Pr[b = 1] + \Pr[\text{FP} \mid b = 0] \Pr[b = 0]\\ 
    &\leq p + (\varepsilon^{\prime} + \negl[\lambda]) (1 - p)\\
    &\leq \varepsilon^{\prime} + p (1 - \varepsilon^{\prime}) + \negl[\lambda]\\
    &\leq \varepsilon^{\prime} + \frac{\varepsilon - \varepsilon^{\prime}}{1 - \varepsilon^{\prime}} (1 - \varepsilon^{\prime}) + \negl[\lambda]\\
    &\leq \varepsilon + \negl[\lambda]
    \end{align*}

Thus, $\mathbf{B}^{\prime}$ is $(n, t, \varepsilon)$-secure in \texttt{LFAGame}. All that is left to show is that $\mathbf{B}^{\prime}$ is not $(n, t, \varepsilon)$-secure in \texttt{LBPGame}. Consider an adversary $\adv^{\prime}$ in \texttt{LBPGame}. Recall that $\adv^{\prime}$ has oracle access to $\mathbf{B}^{\prime}$'s internal classical representation $M^{\prime}$. Once $\mathbf{B}^{\prime}$ is constructed, $\adv^{\prime}$ can read $b$ in $M^{\prime}$ to check whether $\mathbf{B}^{\prime}$ is in the always-one state. $\adv^{\prime}$  chooses to bet only if $b = 1$. Thus, we have:

\begin{align*}
    \expect{C_{\adv^{\prime}}} &= \expect{C_{\adv^{\prime}} \mid b = 1} \Pr[b = 1] + \underbrace{\expect{C_{\adv^{\prime}} \mid b = 0} \Pr[b = 0]}_{0 \text{ since } \adv^{\prime} \text{ won't bet}}\\
    &= \expect{C_{\adv^{\prime}} \mid b = 1} \Pr[b = 1] = {\varepsilon}^{-1} p \geq p
\end{align*}
Therefore, $\expect{C_{\adv^{\prime}}}$ is not negligible, violating the $(n, t, \varepsilon)$-security in \texttt{LBPGame}.\qed
\end{proof}
Figure~\ref{fig:securityimplications} summarizes the relationships that we proved in this section.

\begin{figure}[t!]
\centering
\begin{tikzpicture}[
    node distance=1.25cm and 3.5cm,
    >=Stealth,
    text node/.style={draw, align=center, font=\bfseries, minimum width=1.8cm, minimum height=0.8cm},
    arrow label/.style={font=\small},
    implies/.style={->, double distance=1.5pt, shorten <=2pt, shorten >=2pt},
    notimplies/.style={->, double distance=1.5pt, shorten <=2pt, shorten >=2pt},
    cross/.style={draw=red, thick, line width=1pt}
  ]

  \node (fullyadaptive)     [text node]                              {\texttt{LFAGame}};
  \node (betorpass)         [text node, right=of fullyadaptive]      {\texttt{LBPGame}};
  \node (partiallyadaptive) [text node, below=of fullyadaptive]      {\texttt{LPAGame}};
  \node (alwaysbet)         [text node, below=of betorpass]          {\texttt{LABGame}};

  \draw [implies] ([yshift=-10pt]fullyadaptive.east) -- node[pos=0.3, above, sloped, arrow label] {} ([xshift=-10pt]alwaysbet.north);

  \draw [implies] ([xshift=-8pt]fullyadaptive.south) -- node[pos=0.3, left, arrow label] {} ([xshift=-8pt]partiallyadaptive.north);

  \draw [notimplies] ([xshift=8pt]partiallyadaptive.north) -- node[pos=0.3, right, arrow label] {} node[pos=0.5, sloped, font=\huge\bfseries\color{red}] {$\mathbf{\times}$} ([xshift=8pt]fullyadaptive.south);

  \draw [implies] ([xshift=10pt]betorpass.south) -- node[pos=0.3, right, arrow label] {} ([xshift=10pt]alwaysbet.north);

  \draw [implies] ([yshift=5pt]betorpass.west) -- node[pos=0.3, below, arrow label] {} ([yshift=5pt]fullyadaptive.east);

  \draw [notimplies] ([yshift=10pt]alwaysbet.west) -- node[pos=0.3, sloped, below, xshift=-5pt, font=\small] {} node[pos=0.3, sloped, font=\huge\bfseries\color{red}] {$\mathbf{\times}$} ([xshift=10pt]fullyadaptive.south);

  \draw [notimplies] ([yshift=5pt]partiallyadaptive.east) -- node[pos=0.3, above, arrow label] {} node[pos=0.3, sloped, font=\huge\bfseries\color{red}] {$\mathbf{\times}$} ([yshift=5pt]alwaysbet.west);

  \draw [notimplies] ([yshift=-5pt]alwaysbet.west) -- node[pos=0.3, below, arrow label] {} node[pos=0.3, sloped, font=\huge\bfseries\color{red}] {$\mathbf{\times}$} ([yshift=-5pt]partiallyadaptive.east);

  \draw [notimplies] ([yshift=-5pt]fullyadaptive.east) -- node[pos=0.3, below, arrow label] {} node[pos=0.5, sloped, font=\huge\bfseries\color{red}] {$\mathbf{\times}$} ([yshift=-5pt]betorpass.west);

  \draw [notimplies] ([xshift=-6pt]alwaysbet.north) -- node[pos=0.3, left, arrow label] {} node[pos=0.5, sloped, font=\huge\bfseries\color{red}] {$\mathbf{\times}$} ([xshift=-6pt]betorpass.south);
  
\end{tikzpicture}
\caption{Security notion implications.}
\label{fig:securityimplications}
\end{figure}

%% file: tex/construction.tex
\section{Our Constructions}\label{sec:constructions}
We propose two adaptively-secure LBF constructions; \bloomname{} and \betpassbloomname{}. In both of these constructions, we employ a Partitioned LBF~\cite{plbf} with a partition of cardinality $2$. \bloomname{} combines partitioning with PRPs, while \betpassbloomname{} combines partitioning with Cuckoo hashing and pseudorandom functions (PRFs).\medskip

\noindent\textbf{Construction 1: \bloomname{}.} As shown in Figure~\ref{fig:downtownbodega}, in \bloomname{} $\mathbf{B}$, the learning model $\mathscr{L}$ and threshold $\tau$ are used to partition $S \subseteq \domain$ into two sets: $S_{1} = \{x \in S \mid \mathscr{L}(x) \geq \tau  \}$ and $S_{2} = S \setminus S_{1} = \{x \in S \mid \mathscr{L}(x) < \tau  \}$. We then use  $\prp_{\sk_{A}}$ and $\prp_{\sk_{B}}$ as bijections on sets $S_{1}$ and $S_{2}$, respectively, to form sets $S_{A} = \{ \prp_{\sk_{A}}(x) : x \in S_{1} \}$ and $S_{B} = \{ \prp_{\sk_{B}}(x) : x \in S_{2}\}$. \bloomname{} has two backup CBFs, $\mathbf{B}_{A}$ that stores $S_{A}$ and $\mathbf{B}_{B}$ that stores $S_{B}$. To query \bloomname{} over an element $x \in \domain$, we first evaluate the learning model over $x$. If $\mathscr{L}(x) \geq \tau$, then we compute $y = \prp_{\sk_{A}}(x)$ and pass $y$ to the query algorithm of $\mathbf{B}_{A}$ to obtain an answer for the membership query on whether $x \in S$. On the other hand, if $\mathscr{L}(x) < \tau$, then we repeat the same steps but while using $\prp_{\sk_{B}}$ and the query algorithm of $\mathbf{B}_{B}$.

More formally, \bloomname{} is a data structure with six components $\mathbf{P} = (\mathbf{B}_{A}, \mathbf{B}_{B}, \mathbf{B}_{3}, \mathbf{B}_{4}, \prp_{\sk_{A}}, \prp_{\sk_{B}})$. $\mathbf{B}_{A} = (\mathbf{B}_{{A}_{1}}, \mathbf{B}_{{A}_{2}})$ and $\mathbf{B}_{B} = (\mathbf{B}_{{B}_{1}}, \mathbf{B}_{{B}_{2}})$ are backup CBFs (act like \securecbf{}). $\mathbf{B}_{3}$ is a randomized dataset construction algorithm that on input $S$ constructs a training dataset $\mathscr{T}$ for $S$. $\mathbf{B}_{4}$ is a randomized learning model construction algorithm that on input $\mathscr{T}$ outputs a learning model $\mathscr{L}$ and a threshold $\tau \in [0, 1]$. $\prp_{\sk_{A}}$ and $\prp_{\sk_{B}}$ are pseudorandom permutations with secret keys $\sk_{A}$ and $\sk_{B}$. The internal representation of $S$ consists of:
 \begin{enumerate}
     \item ${M}_{A}$, the representation of $S_{A} = \{\prp_{\sk_{A}}(x) : x \in S \mid \mathscr{L}(x) \geq \tau \}$ stored by backup CBF $\mathbf{B}_{A}$.
     \item ${M}_{B}$, the representation of $S_{B} = \{\prp_{\sk_{B}}(x) : x \in S \mid \mathscr{L}(x) < \tau \}$ stored by backup CBF $\mathbf{B}_{B}$.
     \item $(\mathscr{L}, \tau)$, the learning model and the threshold.
 \end{enumerate}

The query algorithm for \bloomname{} is $\mathbf{B}_{2}(M_{A},M_{B},\mathscr{L}, \tau, x) = (\mathscr{L}(x) \geq \tau 
\land \mathbf{B}_{{A}_{2}}({M_{A}}, \prp_{\sk_{A}}(x)) = 1)
\lor (\mathscr{L}(x) < \tau 
\land \mathbf{B}_{{B}_{2}}(M_{B}, \prp_{\sk_{B}}(x)) = 1)$. Similar to prior work~\cite{moni1,naor2022bet,filic1}, we assume that the internal state available to the adversary does \textbf{not} include the PRP secret keys, which are held securely.\medskip

\begin{figure}[t!]
  \centering
  \begin{tikzpicture}[
    node distance=0.5cm,
    lmbox/.style={draw, minimum width=1cm, minimum height=0.5cm, align=center},
    prpbox/.style={draw, minimum width=0.5cm, minimum height=0.5cm, align=center},
    bfbox/.style={draw, minimum width=1.0cm, minimum height=0.5cm, align=center},
    arrow/.style={-{Stealth[length=3mm]}, thick},
    label/.style={font=\small}
]

\node[lmbox] (LM) {Learning Model};
\node[prpbox, right=1.8cm of LM] (PRP1) {PRP};
\node[bfbox, right=0.5cm of PRP1] (BFB) {CBF $\mathbf{B}_{B}$};

\node[prpbox, below=0.8cm of LM] (PRP2) {PRP};

\node[bfbox, below=0.5cm of PRP2] (BFA) {CBF $\mathbf{B}_{A}$};

\draw[arrow] ([xshift=-1.5cm]LM.west) -- node[label, above] {Input} (LM.west);

\draw[arrow] (LM.east) -- node[label, above] {Negatives} (PRP1.west);

\draw[arrow] (PRP1.east) -- (BFB.west);

\draw[arrow] (LM.south) -- node[pos=0.3, label, right] {Positives} (PRP2.north);

\draw[arrow] (PRP2.south) -- (BFA.north);

\draw[arrow] (BFB.east) -- +(1.5cm, 0cm) node[pos=0.5, label, above] {Negatives};

\draw[arrow] (BFB.south) -- +(0, -0.8cm) node[pos=0.3, label, right] {Positives};

\draw[arrow] (BFA.west) -- +(-2.0cm, 0) node[pos=0.4, label, above] {Negatives};

\draw[arrow] (BFA.east) -- +(2.0cm, 0) node[pos=0.4, label, above] {Positives};

\end{tikzpicture}
  \caption{The \bloomname{} construction.}\label{fig:downtownbodega}
\end{figure}

\noindent\textbf{Construction 2: \betpassbloomname{}.} This construction, shown in Figure~\ref{fig:cuckoo-lbf}, combines partitioning with Cuckoo hashing and PRFs. We first review prior constructions that also used Cuckoo hashing, which form the basis for ours. Naor et al.~\cite{naor2022bet,lotan2025adversarially} present a CBF construction, that is provably secure under \texttt{BPGame}, based on a prior Cuckoo hashing-based construction by Naor and Yogev~\cite{moni1}. We denote this construction \text{Naor-Oved-Yogev} Cuckoo filter or simply \betpasscbf{}. \betpasscbf{} is similar to a CBF variant called a Cuckoo filter~\cite{fan2014cuckoo}. For a set $S \subseteq \domain$ of size $n$, \betpasscbf{} $\mathbf{B} = (\mathbf{B}_{1}, \mathbf{B}_{2})$ stores $S$ using two tables \( Z_{1}, Z_{2} \), each with \( n_{c} = \bigO{n} \) cells. Each table has a corresponding hash function, denoted as $h_1, h_2 : \domain \to [n_{c}]$, respectively. There is a fingerprint function, namely, a keyed pseudorandom function $\prf_{\sk} : \domain \times \{1,2\} \to \bin^r$, where $r \in \bigO{\log \tfrac{1}{\varepsilon}}$ for target FPR $\varepsilon$. The \betpasscbf{} works as follows:
\begin{itemize}
    \item Construction algorithm $\mathbf{B}_{1}$: Stores $S$ in a Cuckoo hashing dictionary~\cite{pagh2001cuckoo}, where an element $x$ is stored in either $Z_{1}[h_{1}(x)]$ or $Z_{2}[h_{2}(x)]$. To save space, the PRF output is stored instead of $x$. In particular, if $x$ is to be stored in $Z_{1}$, we store $y_{1} = \prf_{\sk}(x, 1)$. Otherwise, we store $y_{2} = \prf_{\sk}(x, 2)$ in $Z_2$.

    \item Query algorithm $\mathbf{B}_2$: To query \( x \in \domain \), we compare \(\prf_{\sk}(x, 1)\) with $Z_{1}[h_1(x)]$ and \(\prf_{\sk}(x, 1)\) with \(Z_{2}[h_{2}(x)]\). If either match occurs, we return 1.
\end{itemize}

\betpasscbf{} achieves the completeness and soundness properties from Definition~\ref{def:naoryogevbloom}. However, correctness is guaranteed only if the filter parameters are tuned carefully. Boskov et al.~\cite{boskov2020birdwatching} show empirically that false negatives occur even in state-of-the-art Cuckoo filter implementations.

We combine \betpasscbf{} with the partitioning strategy we used for \bloomname{} to obtain \betpassbloomname{}. Formally, \betpassbloomname{} is a data structure with seven components $\mathbf{C} = (\mathscr{Z}, \mathcal{H}, \mathbf{B}_{2}, \mathbf{B}_{3}, \mathbf{B}_{4}, \prf_{\sk_{A}}, \prf_{\sk_{B}})$. $\mathscr{Z} = \{ Z_{1}, \cdots, Z_{4}\}$ is a set of 4 tables, and $\mathcal{H} = \{ h_{1}, \cdots, h_{4} \}$ is a set of 4 hash functions. $\mathbf{B}_{2}, \mathbf{B}_{3}, \mathbf{B}_{4}$ are polynomial time algorithms. $\mathbf{B}_{2}$ is a query algorithm, $\mathbf{B}_{3}$ is a randomized dataset construction algorithm, and $\mathbf{B}_{4}$ is a randomized learning model construction algorithm (with the same description as before). $\prf_{\sk_{A}}$ and $\prf_{\sk_{B}}$ are PRFs with secret keys $\sk_{A}$ and $\sk_{B}$, respectively. The internal representation of a \betpassbloomname{} storing a set $S$ consists of the following:
\begin{enumerate}
     \item ${M}_{A}$, the representation of \(S_{A} = \{x \in S \mid \mathscr{L}(x) \geq \tau \}\) stored by \betpasscbf{} $\textbf{B}_{A}$, including tables $Z_{1}, Z_{2}$, hash functions $h_{1}, h_{2}$, and $\prf_{\sk_{A}}$. 
     \item ${M}_{B}$, the representation of \(S_{B} = \{x \in S \mid \mathscr{L}(x) < \tau \}\) stored by \betpasscbf{} $\mathbf{B}_{B}$, including tables $Z_{3}, Z_{4}$, hash functions $h_{3}, h_{4}$, and $\prf_{\sk_{B}}$. 
     \item $(\mathscr{L}, \tau)$, the learning model and the threshold.
\end{enumerate}

The query algorithm $\mathbf{B}_{2}$ for a \betpassbloomname{} is \( \mathbf{B}_{2}(M_{A},M_{B},\mathscr{L}, \tau, x) = (\mathscr{L}(x) \geq \tau \land \mathbf{B}_{{A}_{2}}({M}_{A}, x) = 1) \lor (\mathscr{L}(x) < \tau \land \mathbf{B}_{{B}_{2}}(M_{B}, x) = 1) \) where $\mathbf{B}_{{A}_{2}}$ and $\mathbf{B}_{{B}_{2}}$ are query algorithms for the \betpasscbf{}s $\mathbf{B}_{A}$ and $\mathbf{B}_{B}$, respectively. As before, we assume that the internal state available to the adversary does \textbf{not} include the PRF secret keys.\medskip

\begin{figure}[t!]
  \centering
\begin{tikzpicture}[
    node distance=1cm,
    lmbox/.style={draw, minimum width=2cm, minimum height=0.5cm, align=center},
    tablebox/.style={draw, minimum width=4.0cm, minimum height=1.0cm, align=center},
    bigtablebox/.style={draw, minimum width=4.5cm, minimum height=2.8cm, align=center},
    cell/.style={draw, minimum width=0.8cm, minimum height=0.4cm},
    arrow/.style={-{Stealth[length=2mm]}, thick},
    dashedarrow/.style={-{Stealth[length=2mm]}, thick, dashed}
]

\node[lmbox] (LM) {Learning Model};

\draw[arrow] ([xshift=-1.2cm]LM.west) -- node[above] {Input} (LM.west);

\node[bigtablebox, right=2.2cm of LM.east, yshift=-0.1cm] (BigT3) {};

\node[tablebox, right=2.5cm of LM.east, yshift=0.6cm] (T3) {};
\node at (T3.north west) [anchor=north west, xshift=0.1cm, yshift=-0.1cm] {$\mathrm{Z}_3$};

\node[cell] at ($(T3.west) + (0.6,-0.2)$) {};
\node[cell] at ($(T3.west) + (1.5,-0.2)$) {};
\node[cell] at ($(T3.west) + (2.4,-0.2)$) {};
\node[cell, align=center, font=\tiny] at ($(T3.west) + (3.4,-0.2)$) (PRF1) {$\prf_{\sk}(x)$};

\node[tablebox, below=0.3cm of T3] (T4) {};
\node at (T4.north west) [anchor=north west, xshift=0.1cm, yshift=-0.1cm] {$\mathrm{Z}_4$};

\node[cell] at ($(T4.west) + (0.6,-0.2)$) {};
\node[cell] at ($(T4.west) + (1.5,-0.2)$) {};
\node[cell, align=center, font=\tiny] at ($(T4.west) + (2.5,-0.2)$) (PRF2) {$\prf_{\sk}(x)$};
\node[cell] at ($(T4.west) + (3.5,-0.2)$) {};

\node[bigtablebox, below=0.65cm of LM, yshift=-0.1cm] (BigT1) {};
\node[tablebox, below=1cm of LM] (T1) {};
\node at (T1.north west) [anchor=north west, xshift=0.1cm, yshift=-0.1cm] {$\mathrm{Z}_1$};

\node[cell] at ($(T1.west) + (0.6,-0.2)$) {};
\node[cell] at ($(T1.west) + (1.5,-0.2)$) {};
\node[cell] at ($(T1.west) + (2.4,-0.2)$) {};
\node[cell, align=center, font=\tiny] at ($(T1.west) + (3.4,-0.2)$) {$\prf_{\sk}(x)$};

\node[tablebox, below=0.3cm of T1] (T2) {};
\node at (T2.north west) [anchor=north west, xshift=0.1cm, yshift=-0.1cm] {$\mathrm{Z}_2$};

\node[cell, align=center, font=\tiny] at ($(T2.west) + (0.6,-0.2)$) {$\prf_{\sk}(x)$};
\node[cell] at ($(T2.west) + (1.6,-0.2)$) {};
\node[cell] at ($(T2.west) + (2.5,-0.2)$) {};
\node[cell] at ($(T2.west) + (3.4,-0.2)$) {};

\draw[arrow] (LM.east) -- node[above] {Negatives} ($(LM.east)+(1.8,0)$);
\draw[dashedarrow] ($(LM.east)+(1.8,0)$) -- (PRF1.west);
\draw[dashedarrow] ($(LM.east)+(1.8,0)$) -- (PRF2.north);

\draw[arrow] (LM.south) -- node[right] {Positives} (T1.north);

\node[font=\tiny] at ($(T3.south)!0.5!(T4.north) + (-0.6,-0.0)$) {$\mathrm{h}_3(\mathrm{x})$};
\node[font=\tiny] at ($(T4.north west) + (2.0,-0.15)$) {$\mathrm{h}_4(\mathrm{x})$};
\node[font=\tiny] at ($(T1.south)!0.5!(T1.north) + (1.2,0.3)$) {$\mathrm{h}_1(\mathrm{x})$};
\node[font=\tiny] at ($(T1.south)!0.5!(T2.north) + (-1.2,0)$) {$\mathrm{h}_2(\mathrm{x})$};

\draw[dashedarrow] ($(T1.west) + (2.0,0.45)$) -- ($(T2.west) + (0.8,0)$);
\draw[dashedarrow] ($(T1.west) + (2.0,0.45)$) -- ($(T1.west) + (3.2,0)$);

\draw[arrow] ([xshift=1cm]BigT3.south) -- node[pos=0.7, right] {Negatives} ++(0,-0.8cm);
\draw[arrow] ([xshift=-1cm]BigT3.south) -- node[pos=0.7, right] {Positives} ++(0,-0.8cm);

\draw[arrow] ([yshift=0cm]BigT1.east) -- node[pos=0.75, above] {Positives} ++(1.2cm,0);
\draw[arrow] ([yshift=-1.2cm]BigT1.east) -- node[pos=0.75, above] {Negatives} ++(1.2cm,0);

\end{tikzpicture}
  \caption{The \betpassbloomname{} construction.} \label{fig:cuckoo-lbf}
\end{figure}

\noindent\textbf{False Positive Rate Analysis.}
We analyze FPR of both constructions in non-adversarial settings. We show that for \bloomname{} $\mathbf{P}$; the same analysis and results hold for \betpassbloomname{} $\mathbf{C}$ as both constructions use the same partitioning strategy. We also note that in our analysis, FPR refers to the probability of the event that some input $x \in \domain$ is a FP. Thus, in this and other sections that analyze FPR, we compute the probability of this event. A false positive (FP) for a query $x$ happens in \bloomname{} if any of the following holds (see Figure~\ref{fig:hybridmodel}): 
\begin{enumerate}
    \item $x$ generates a FP in $\mathscr{L}$ and a FP in $\mathbf{B}_{A}$. 
    \item $x$ generates a TN (true negative) in $\mathscr{L}$ and a FP in  $\mathbf{B}_{B}$.
\end{enumerate} 

\begin{figure}[t!]
    \centering
\begin{tikzpicture}[
    node distance=0.5cm,
    lmbox/.style={draw, minimum width=1cm, minimum height=0.5cm, align=center},
    prpbox/.style={draw, minimum width=0.5cm, minimum height=0.5cm, align=center},
    bfbox/.style={draw, minimum width=1.0cm, minimum height=0.5cm, align=center},
    arrow/.style={-{Stealth[length=3mm]}, thick},
    reddarrow/.style={-{Stealth[length=2mm]}, red, dashed},
    label/.style={font=\small},
    redlabel/.style={font=\small, red}
]

\node[lmbox] (LM) {Learning Model};
\node[prpbox, right=1.8cm of LM] (PRP1) {PRP};
\node[bfbox, right=0.5cm of PRP1] (BFB) {CBF $\mathbf{B}_{B}$};

\node[prpbox, below=0.8cm of LM] (PRP2) {PRP};

\node[bfbox, below=0.5cm of PRP2] (BFA) {CBF $\mathbf{B}_{A}$};

\draw[arrow] ([xshift=-1.5cm]LM.west) -- node[label, above] {Input} (LM.west);

\draw[reddarrow] (-2.5cm,0.8cm) -- ++(2cm,0.0cm) -| (LM.north)
    node[pos=-0.75, redlabel, above] {Adversarial input};

\draw[arrow] (LM.east) -- node[label, above] {Negatives} (PRP1.west);

\draw[reddarrow] ([xshift=0.3cm]LM.north) -- ++(0cm,0.5cm) -| (PRP1.north) 
    node[pos=0.25, redlabel, above] {True negatives};

\draw[arrow] (PRP1.east) -- (BFB.west);
\draw[reddarrow] ([xshift=0.3cm]PRP1.north) -- ++(0cm,0.5cm) -| ([xshift=-0.3cm]BFB.north);

\draw[arrow] (LM.south) -- node[pos=0.3, label, right] {Positives} (PRP2.north);

\draw[reddarrow] ([xshift=-0.3cm]LM.south) -- node[pos=0.5, redlabel, left] {False positives} ([xshift=-0.3cm]PRP2.north);

\draw[arrow] (PRP2.south) -- (BFA.north);
\draw[reddarrow] ([xshift=-0.3cm]PRP2.south) -- node[pos=0.5, redlabel, left] {} ([xshift=-0.3cm]BFA.north);

\draw[arrow] (BFB.north) -- +(0, 0.8cm) node[pos=0.5, label, right] {Negatives};

\draw[arrow] (BFB.south) -- +(0, -0.8cm) node[pos=0.3, label, right] {Positives};
\draw[reddarrow] ([xshift=-0.3cm]BFB.south) -- ++(0.0cm, -0.8cm) node[pos=0.5, redlabel, left] {False positives};

\draw[arrow] (BFA.west) -- +(-2.0cm, 0) node[pos=0.4, label, above] {Negatives};

\draw[arrow] (BFA.east) -- +(2.0cm, 0) node[pos=0.4, label, above] {Positives};
\draw[reddarrow] ([yshift=-0.2cm]BFA.east) -- ++(1.0cm, 0) node[pos=0.4, redlabel, above, yshift=0.5cm] {};

\end{tikzpicture}
    \caption{To generate a false positive in a \bloomname{}, an adversary $\adv$ must either generate a false positive in the learning model and direct its query through backup CBF $\mathbf{B}_{A}$, or generate a true negative in the learning model and direct its query through backup CBF $\mathbf{B}_{B}$}\label{fig:hybridmodel}
\end{figure}

For $x \sample \domain$,\footnote{The quantification of FPR here, or simply the probability a given input is FP, is for non-adversarial queries, hence, we have $x \sample \domain$.} let $\text{FP}(x, S, m)$ denote the event that $x$ is a false positive in a CBF that encodes the set $S$ with memory budget $m$, $\text{FP}_{L}(x, S, \mathscr{T}, m)$ denote the event that $x$ is a false positive in a learning model $\mathscr{L}$ generated using the training dataset $\mathscr{T}$ (where this dataset is generated based on $S$) with memory budget $m$, and $\text{TN}_{L}(x, S, \mathscr{T}, m)$ denote the event that $x$ is a true negative in a learning model $\mathscr{L}$ generated using the training dataset $\mathscr{T}$ (also based on the set $S$) with memory budget $m$. We assume that the correctness probability of the learning model is independent of that of the backup CBFs. In particular, we assume that for any $m, m^{\prime}$, we have: 
\[\Pr[\text{FP}_{L}(x, S, \mathscr{T}, m) \cap \text{FP}(x, S, m^{\prime})] = \Pr[\text{FP}_{L}(x, S, \mathscr{T}, m)] \, \Pr[\text{FP}(x, S, m^{\prime})]\] 
\[\Pr[\text{TN}_{L}(x, S, \mathscr{T}, m) \cap \text{FP}(x, S, m^{\prime})] = \Pr[\text{TN}_{L}(x, S, \mathscr{T}, m)] \, \Pr[\text{FP}(x, S, m^{\prime})]\]

Consider a total memory budget of $m_{T}$. The memory allocation of a \bloomname{} from $m_{T}$ is assigned as follows: $m_{L}$ bits are for the learning model $\mathscr{L}$, $m_{A}$ bits are for backup CBF $\mathbf{B}_{A}$, $m_{B}$ bits are for backup CBF $\mathbf{B}_{B}$, $2\lambda$ bits are for the PRP keys. Thus, we have $m_{T} \geq m_{L} + m_{A} + m_{B} + 2\lambda$. Based on that, we have the following theorem. 

\begin{theorem}\label{thm:fprdowntownbodegafilter}
Let $x \sample \domain$, for any memory budget $m_{T} \in \NN$, set $S \subseteq \domain$, let $\text{FP}_{LBF}$ denote the event that $x$ is a false positive in \bloomname{} (respectively, \betpassbloomname{}) encoding $S$. The probability of this event (in non-adversarial settings) is:
\begin{align*}
\Pr[\text{FP}_{LBF}(x, S, \mathscr{T}, m_{T})] =\, &\Pr[\text{FP}_{L}(x, S, \mathscr{T}, m_{L})] \, \Pr[\text{FP}(x, S_{A}, m_{A})] \\ & + \Pr[\text{TN}_{L}(x, S, \mathscr{T}, m_{L})] \, \Pr[\text{FP}(x, S_{B}, m_{B})]    
\end{align*}

\noindent where $\mathscr{T}$ and $\mathscr{L}$ are the training dataset and learning model, respectively, corresponding to $S$,  \(S_{A} = \{x \in S \mid \mathscr{L}(x) \geq \tau \}\), \(S_{B} = \{x \in S \mid \mathscr{L}(x) < \tau \}\), and the probability is taken over the random coins of $\mathscr{L}$ and the backup CBFs used in \bloomname{} (respectively, \betpassbloomname{}).

\end{theorem}

\begin{proof}
    Based on the design of \bloomname{} (and \betpassbloomname{}), it follows that a resulting FP must either be a FP in the backup CBF encoding $S_{A}$ or in the backup CBF encoding $S_{B}$. A FP in the former must also be a FP in $\mathscr{L}$, while a FP in the latter must be a TN in $\mathscr{L}$.\qed
\end{proof}

%% file: tex/sec-perf-analysis.tex
\section{Security and Performance Analysis}

In this section, we formally prove security of our LBF constructions. We show that \bloomname{} is secure under \texttt{LABGame}, \texttt{LFAGame}, and \texttt{LPAGame}, and that \betpassbloomname{} is secure under the \texttt{LBPGame} and \texttt{LABGame}. We leave studying whether \betpassbloomname{} is secure under $\mathtt{LFAGame}$ and $\mathtt{LPAGame}$ or introducing a construction that is secure against all LBF games as a future work direction. 

We also analyze FPR (achieved for a given memory budget) of our constructions under these adversarial settings. Our results show that our LBF constructions achieve strong security guarantees while maintaining competitive performance (in terms of FPR achieved for a given memory budget). 

\subsection{Fully Adaptive Adversary}

This section establishes the security of \bloomname{} against fully adaptive adversaries. We first prove that \securecbf{} maintains its security guarantees even when facing a fully adaptive adversary who can access the filter's internal representation. Building on this result, we then prove that \bloomname{} preserves these security properties in the LBF context.

\begin{theorem}\label{thm:fullyadaptivecbf}
    Let $\mathbf{B}$ be an $(n, \varepsilon)$ \securecbf{}. Assuming PRPs exist, then for a security parameter $\lambda$ there exists a negligible function $\negl[\cdot]$ such that $\mathbf{B}$ is $(n, t, \varepsilon + \negl)$-secure under \texttt{FAGame} for any $t \in \bigO{\poly[n, \lambda]}$.
\end{theorem}
\begin{proof}
    By Theorem~\ref{thm:moninaortheorem}, we know that \securecbf{} is a correct and sound CBF that is $(n, t, \varepsilon + \negl)$-secure for any $t \in \bigO{\poly[n, \lambda]}$ under \texttt{ABGame}. All that is left to show is that the additional oracle access to $\oracle = \{\mathbf{B}_1 (\cdot), \mathbf{B}_2 (M, \cdot)\}$ in \texttt{FAGame} does not allow a $\ppt$ adversary $\adv$ to win in \texttt{FAGame} with non-negligible probability. We can prove this using a hybrid argument. Let $\mathbf{B}$ be a \securecbf{}.

    \text{Hybrid Game \texttt{H0Game}}: This is the original game with $\mathbf{B}$.
    
    \text{Hybrid Game \texttt{H1Game}}: Let $\mathbf{B}^{\prime}$ be the same construction as $\mathbf{B}$ but with $\prp_{\sk}$ replaced with a truly random permutation $\pi$. 
    
    Since $\prp$ is a secure PRP, i.e., indistinguishable from a truly random permutation, thus $\adv$ cannot distinguish between \texttt{H0Game} and \texttt{H1Game}. In \texttt{H1Game}, to $\adv$, the representation $\mathbf{B}_{1}(\pi(S)) = M_{\pi}$ is indistinguishable from a representation constructed from a random set. So $\adv$ cannot gain any information about set $S$ from $\mathbf{B}_{1}(\pi(S))$ and $\adv$'s view is identical to that in \texttt{ABGame}. Hence,
    \(\Pr[\adv \text{ wins } \texttt{H1Game}] \leq \Pr[\adv \text{ wins } \texttt{ABGame}] = \varepsilon + \negl
    \). 
    
    Now, in the original construction $\mathbf{B}$ that uses $\prp_{\sk}$, assuming PRPs exist, we have \( 
    \Pr[\adv \text{ wins } \texttt{FAGame}] = \Pr[\adv \text{ wins } \texttt{H1Game}] + \negl
  = \varepsilon + \negl \). Therefore, \securecbf{} is $(n, t, \varepsilon + \negl)$-secure under \texttt{FAGame}.\qed
\end{proof}

The following theorem shows that a \bloomname{} is $(n, t, \varepsilon + \negl)$-secure under \texttt{LFAGame}. To prove this result, we first prove that the differences between a \bloomname{} and a standard LBF construction still allow a \bloomname{} to be a correct  $(n, \tau, \varepsilon, \varepsilon_{p}, \varepsilon_{n})$-LBF. Put differently, if there exists a standard LBF construction for any set $S \subseteq \domain$ of cardinality $n$ that satisfies the properties of an $(n, \tau, \varepsilon, \varepsilon_{p}, \varepsilon_{n})$-LBF, then there also exists a \bloomname{} construction that satisfies those properties. We prove security by doing a case analysis that reduces the security of \bloomname{} under \texttt{LFAGame} to the security of \securecbf{} under \texttt{FAGame}.

\begin{theorem}\label{thm:downtownbodegafilter}
    Let $\mathbf{B}$ be a standard construction for an $(n, \varepsilon, \varepsilon_{p}, \varepsilon_{n})$-LBF that uses $m$ bits of memory out of which $m_{C}$ bits are used for the backup CBF. Assuming PRPs exist, then for security parameter $\lambda$ there exists a negligible function $\mathrm{negl(\cdot)}$ and a \bloomname{}, $\mathbf{B}_{\sk_{A}, \sk_{B}}$, that is $(n, t, \varepsilon + \negl)$-secure for any $t \in \bigO{\poly[n,\lambda]}$ under \texttt{LFAGame}, and uses $m^{\prime} = m + m_{C} + 2 \lambda$ bits of memory.
\end{theorem}

\begin{proof}
 Standard LBF $\mathbf{B}$ can be transformed into a \bloomname{} $\mathbf{P}$ as follows. Choose random secret keys $\sk_{A}, \sk_{B} \in \{0, 1\}^{\lambda}$ and use $2\lambda$ bits of extra memory to store them. Use the memory budget of $\mathbf{B}$'s backup CBF to construct backup CBF $\mathbf{B}_{A}$. Use $m_{C}$ extra bits to construct backup CBF $\mathbf{B}_{B}$. Keep the learning model $\mathscr{L}$ as is. $\mathbf{P}$'s completeness follows from $\mathbf{B}_{B}$'s completeness and the fact that for any $x$ such that $\mathscr{L}(x) < \tau$, $\mathbf{P}$ will return $x \notin S$ if and only if the query algorithm of $\mathbf{B}_{B}$ outputs $0$. The soundness of $\mathbf{B}_{\sk_{A}, \sk_{B}}$ follows from the soundness of $\mathbf{B}_{A}$ and $\mathbf{B}_{B}$, and $\mathbf{P}$'s learning model soundness follows from $\mathscr{L}$'s soundness.

 Consider a false positive (FP), i.e., an $x \notin S$ for which $\mathbf{B}_{\sk_{A}, \sk_{B}}$ returns $1$. This occurs in two cases:

\begin{itemize}
\item \text{Case 1}: $\mathscr{L}(x) \geq \tau$ and $\mathbf{B}_{A}$ returns $1$, i.e., ${\mathbf{B}_{A}}_{2}(M_{A}, \prp_{\sk_{A}}(x)) = 1$.
    
\item \text{Case 2}: $\mathscr{L}(x) < \tau$ and $\mathbf{B}_{B}$ returns $1$, i.e., ${\mathbf{B}_{B}}_{2}({M_{B}}, \prp_{\sk_{B}}(x)) = 1$.

\end{itemize}

In both cases, for adversary $\adv$ to induce a FP in $\mathbf{P}$, it must induce a FP either in backup CBF $\mathbf{B}_{A}$ or in backup CBF $\mathbf{B}_{B}$. $\mathbf{B}_{A}$ and $\mathbf{B}_{B}$ are \securecbf{}s and, by Theorem~\ref{thm:fullyadaptivecbf}, are $(n, t, \epsilon + \negl)$-secure under \texttt{FAGame}. Therefore, $\mathbf{P}$ is $(n, t, \epsilon + \negl)$-secure under \texttt{LFAGame}.\qed
\end{proof}

We now discuss how \bloomname{} mitigate concrete attacks discussed in the literature, namely, the two attacks on LBFs by Reviriego et al.~\cite{reviriego1} (which we refer to as opaque-box and clear-box attacks), as well as a general poisoning attack on learned index structures introduced by Kornaropoulos et al.~\cite{kornaropoulos2022price}. 

\textbf{Opaque-box attack.} The opaque-box adversarial model is similar to \texttt{LABGame} as both allow $\adv$ to query the LBF. However, in \texttt{LABGame} $\adv$ chooses the LBF's input set $S$, whereas in Reviriego et al.'s model, $\adv$ does not choose that. The opaque-box attack first tests elements until a false positive or a true positive is found. They it mutates the positive by changing a small fraction of the bits in the input to generate more false positives. The attack targets the learning model in standard LBFs by making it generate false positives without having the input reach the backup CBF. Unlike a standard LBF, \bloomname{} ensures both positive and negative queries are routed to a backup CBF that is $(n, t, \varepsilon)$-secure under \texttt{ABGame}. This ensures that the opaque-box attack will not induce a false positive in \bloomname{} with probability non-negligibly greater than $\varepsilon$.

\textbf{Clear-box attack.} The clear-box adversarial model is similar to \texttt{LFAGame}. With knowledge of the internal state of the LBF's learning model, $\adv$ generates mutations in a more sophisticated way. Reviriego et al. provide the example of a malicious URL dataset where $\adv$ may begin with a non-malicious URL and make changes such as removing the \texttt{"s"} in \texttt{"https"}. Since \bloomname{} is $(n, t, \epsilon)$-secure even when $\adv$ has access to oracle $\oracle = \{\mathbf{B}_1 (\cdot), \mathbf{B}_3(\cdot), \mathbf{B}_4(\cdot)\}$, which reveals learning model state, \bloomname{} remains secure.

\textbf{Poisoning attacks.} Kornaropoulos et al.~\cite{kornaropoulos2022price} discuss an attack where $\adv$ poisons the learning model’s training dataset by having maliciously-chosen inputs in this set. This poisoning attack modifies the training dataset, but not the queries sent to the LBF. The results of our \texttt{LFAGame} hold even if we let $\adv$ choose $\mathscr{T}$ as long as the challenger validates that $\mathscr{T}$ satisfies Definition~\ref{def:trainingdataset}. To accommodate poisoning attacks, we can let $\adv$ choose a $\mathscr{T}$ that is not validated by the challenger. Even in this relaxed version of \texttt{LFAGame}, \bloomname{} will prevent $\adv$ from inducing false positives with probability non-negligibly larger than $\varepsilon$. This is because our security guarantees do not rely on the learning model, but on the $(n, t, \varepsilon)$-secure backup CBFs, which do not use the training dataset $\mathscr{T}$.

\begin{theorem}\label{thm:fullyadaptivefprexpression}
    In \texttt{LFAGame}, for a $\ppt$ adversary $\adv$ that outputs a guess $x^{*} \in \domain$, the probability that $x^{*}$ is a false positive in a \bloomname{} $\textbf{P}$ storing set $S \subseteq \domain$ with training dataset $\mathscr{T}$ and learning model $\mathscr{L}$ is: 
    \[ \Pr[\text{FP}_{LBF}(x^{*}, S, \mathscr{T}, m_{T})] \leq \max(\Pr[\text{FP}(x^{*}, S_{A}, m_{A})], \Pr[\text{FP}(x^{*}, S_{B}, m_{B}))] \]
    where $S_{A} = \{ x \in S \mid \mathscr{L}(x) \geq \tau\}$, $S_{B} = S \setminus S_{A} = \{ x \in S \mid \mathscr{L}(x) < \tau\}$, $\text{FP}_{LBF}$ is the event denoting a false positive in $\textbf{P}$, $\text{FP}$ is the event denoting a false positive in \securecbf{}, $m_T$, $m_A$, and $m_B$ are the total memory of $\textbf{P}$, memory used by backup CBF $\mathbf{B}_A$, and memory used by backup CBF $\mathbf{B}_B$, respectively, and the probability is taken over the random coins of $\adv$ and $\textbf{P}$.
\end{theorem}

\begin{proof}
    As established in Theorem~\ref{thm:fprdowntownbodegafilter}, $x^{*}$ can only induce a FP in $\textbf{P}$ if $x^{*}$ also induces a FP or a TN in the learning model $\mathscr{L}$. Therefore, the probability of $x^{*}$ inducing a FP in $\textbf{P}$ will be the probability of $x^{*}$ inducing a FP in one of the backup CBFs, i.e., the probability will be $\Pr[\text{FP}(x^{*}, S_{A}, m_{A})]$ or $\Pr[\text{FP}(x^{*}, S_{B}, m_{B})]$. Thus, the upper bound above follows.\footnote{Note that if $\adv$ chooses $x^*$ at random from $\domain$, then this reduces to the non-adversarial case analyzed in Theorem~\ref{thm:fprdowntownbodegafilter}. The bound in that theorem also respects the bound stated in Theorem~\ref{thm:fullyadaptivefprexpression} above.} \qed
\end{proof}

\subsection{Partially Adaptive Adversary}

Recall that Theorem~\ref{thm:lfaimplieslpa} proves that any LBF that is $(n, t, \varepsilon)$-secure under \texttt{LFAGame} is $(n, t, \varepsilon)$-secure under \texttt{LPAGame}. Since \bloomname{} is $(n, t, \varepsilon)$-secure under \texttt{LFAGame}, as we proved in the previous section, it is also $(n, t, \varepsilon)$-secure under \texttt{LPAGame}. Note that in both \texttt{LFAGame} and \texttt{LPAGame}, a $\ppt$ adversary $\adv$ outputs a guess $x^{*} \in \domain$. The difference between \texttt{LFAGame} and \texttt{LPAGame} is the fraction $\alpha$ of $\adv$'s initial exploratory query budget $t$. Therefore, for a guess $x^{*} \in \domain$ output by adversary in \texttt{LPAGame}, the probability of $x^{*}$ inducing a FP in \bloomname{} is the same as the expression we derived for \texttt{LFAGame} in Theorem~\ref{thm:fullyadaptivefprexpression}.

As discussed in Section~\ref{sec:partially-adaptive-security}, \texttt{LPAGame} is actually designed to capture is a \emph{mixed workload} where a percentage of the queries are adversarial and the rest are non-adversarial. This is a more relevant scenario when it comes to analyzing real-world performance in terms of FPR for a given memory budget under a given workload. Thus, we analyze FPR over the $t$ queries $x_{1}, \ldots, x_{t}$ in \texttt{LPAGame}, covering the $\alpha t$ adversarial queries and the $(1 - \alpha)t$ non-adversarial queries. For clarity, we refer to these $t$ queries as workload queries to distinguish them from the guess $x^*$. 

Without loss of generality, let $\alpha_{P}$ of the adversarial queries generate FPs in the learning model that go through backup CBF $\mathbf{B}_{A}$. Similarly, let $\alpha_{N}$ of the adversarial queries generate TNs in the learning model that go through backup CBF $\mathbf{B}_{B}$. Note that $\alpha = \alpha_{P} + \alpha_{N}$. The adversary makes at most $\alpha_{P} t$ queries that generate FPs in the learning model and at most $\alpha_{N} t$ queries that generate TNs in the learning model.

\begin{theorem}\label{thm:hybriddowntownbodegaexpression}
    In \texttt{LPAGame}, for a workload query $x_{i} \in \domain$, the probability that $x_{i}$ is a false positive in a \bloomname{} $\textbf{P}$ storing set $S \subseteq \domain$ with training dataset $\mathscr{T}$ and learning model $\mathscr{L}$ is
    \begin{align*}
        \alpha_{P} \Pr[\text{FP}(x_{i}, S_{A}, m_{A})] &+ \alpha_{N} \Pr[\text{FP}(x_{i}, S_{B}, m_{B})]\\
        &+ (1 - \alpha_{P} - \alpha_{N}) \Pr[\text{FP}_{LBF}(x_{i}, S, \mathscr{T}, m_{T})]
    \end{align*}

    \noindent where $S_{A} = \{ x \in S \mid \mathscr{L}(x) \geq \tau\}$, $S_{B} = S \setminus S_{A} = \{ x \in S \mid \mathscr{L}(x) < \tau\}$, $\text{FP}(\cdot)$ is the event denoting a false positive in a CBF, $\text{FP}_{LBF}(\cdot)$ is the event denoting a false positive in $\textbf{P}$, $\alpha_{P}$ is the fraction of $t$ queries chosen by $\adv$ that induce false positives in $\mathscr{L}$, and $\alpha_{N}$ is the fraction of $t$ queries chosen by $\adv$ that induce TNs in $\mathscr{L}$. The probability is taken over the random coins used by $\adv$, $\textbf{P}$, and the generation of the $(1 - \alpha_{P} - \alpha_{N}) t$ non-adversarial queries.
\end{theorem}

\begin{proof}
    One of the following cases holds for any query $x_{i}$ among the $t$ workload queries in \texttt{LPAGame}.
\begin{itemize}
    \item \text{Case 1}: $x_{i}$ is not adversary-generated. Therefore, as established by Theorem~\ref{thm:fprdowntownbodegafilter}, these have $\Pr[\text{FP}_{LBF}(x_{i}, S, \mathscr{T}, m_{T})]$ to be FP. There are $(1 - \alpha_{P} - \alpha_{N}) t$ such queries.

    \item \text{Case 2}: $x_{i}$ is adversary-generated and generates a FP in the learning model $\mathscr{L}(x_{i})$. Since $\mathscr{L}(x_{i})$ generating a FP and $\mathscr{L}(x_{i})$ generating a TN are mutually exclusive events, the probability of $x_{i}$ inducing a FP in $\textbf{P}$ is just the probability of $x_{i}$ inducing a FP in backup CBF $\mathbf{B}_{A}$, i.e, $\Pr[\text{FP}(x_{i}, S_{A}, m_{A})$. There are $\alpha_{P} t$ such queries.

    \item \text{Case 3}: $x_{i}$ is adversary-generated and generates a TN in $\mathscr{L}(x_{i})$. Similar to case 2, we can derive the probability of $x_{i}$ inducing a FP to be $\Pr[\text{FP}(x_{i},S_{B}, m_{B})]$. There are $\alpha_{N} t$ such queries.
\end{itemize}
The statement of the theorem follows by applying the law of total probability.\qed
\end{proof}

\subsection{Bet-or-Pass Adversary}

\securecbf{} has only been shown to be secure under \texttt{ABGame}. Naor and Oved~\cite{naor2022bet} provide compelling reasons for why \securecbf{} may not be secure under \texttt{BPGame}. Whether it is possible to modify \securecbf{} in a way that makes it secure under \texttt{BPGame} is an open problem. On the other hand, \betpasscbf{} \textbf{is} secure under \texttt{BPGame}, proved by Naor and Oved~\cite{naor2022bet}, and we recall this result below.

\begin{theorem}[NOY Theorem]\label{thm:noybloomtheorem}
Assuming one-way functions exist, for any $n \in \NN$, universe of size $n < u$, and $0 < \varepsilon < 1/2$, there exists a Bloom filter that is $(n, \varepsilon)$-strongly resilient in \texttt{BPGame} and uses $\bigO{n \log \frac{1}{\varepsilon} + \lambda}$ bits of memory. There exists a CBF construction $\mathbf{B}'$ (which is \betpasscbf{} mentioned above) where for any constant $0 < \varepsilon < 1/2$, $\mathbf{B}'$ is an $(n, \varepsilon)$-strongly resilient in \texttt{BPGame} and uses $\bigO{n \log \frac{1}{\varepsilon} + \lambda}$ bits of memory.
\end{theorem}

We now show that our \betpassbloomname{} construction is secure under \texttt{LBPGame}. We prove this using a case analysis of all the decisions available to an adversary $\adv$ in \texttt{LBPGame}. Our case analysis shows that all decision paths reduce the security of \betpassbloomname{} under \texttt{LBPGame} to the security of \betpasscbf{} under \texttt{BPGame}.

\begin{theorem}\label{thm:bet_pass_lbf}
    Let $\mathbf{B}$ be a standard $(n, \varepsilon, \varepsilon_{p}, \varepsilon_{n})$-LBF that uses $m$ bits of memory out of which $m_{C}$ bits are used for the backup CBF and $m_{L}$ bits are used for the learning model, such that $m_{C} + m_{L} = m$. Assuming one-way functions exist, for a security parameter $\lambda$, any $n \in \NN$, domain $\domain$ such that $n < |\domain|$, and $0 < \varepsilon < 0.5$ there exists an LBF that is $(n, t, \negl)$-secure under \texttt{LBPGame} for any $t \in \bigO{\poly[n, \lambda]}$ and uses $m^{\prime} = m_{L} + \bigO{n \log \frac{1}{\varepsilon} + \lambda}$ bits of memory.
\end{theorem}
\begin{proof}
By Theorem~\ref{thm:noybloomtheorem}, we know that \betpasscbf{} is $(n, t, \negl)$-secure CBF under \texttt{BPGame} for any $t \in \bigO{\poly[n, \lambda]}$, and uses $\bigO{n \log{\frac{1}{\varepsilon}} + \lambda}$ bits of memory. Recall that in \texttt{LBPGame}, unlike \texttt{BPGame}, $\adv$ has oracle access to the filter construction algorithm $\mathbf{B}_{1}(\cdot)$ which returns the internal representation $M$. Let \texttt{BPGamePlus} be a modified version of \texttt{BPGame} where $\adv$ has oracle access to $\mathbf{B}_{1}(\cdot)$. 

We first show that \betpasscbf{} $\mathbf{B}$ (that uses $\prf_\sk$) is $(n, t, \negl)$-secure under \texttt{BPGamePlus}. We define a hybrid game \texttt{H1Game} in which $\prf_\sk$ is replaced by a true random function $f$. By the security of PRFs, this hybrid is indistinguishable from the original game that uses PRFs. We denote \texttt{H1Game}’s internal representation of \betpasscbf{} as $M^{\prime}$. To $\adv$, the representation $M^{\prime}$ is indistinguishable from a representation constructed from a random set. So $\adv$ cannot gain any information about the input set $S$ from $M^{\prime}$ and $\adv$’s view is identical to that in \texttt{BPGame}. Hence $\Pr[\adv \text{ wins } \mathtt{H1Game}] \leq \Pr[\adv \text{ wins } \mathtt{BPGame}] = \varepsilon + \negl[\lambda]$. Now, in the original construction $\mathbf{B}$ that uses $\prf_\sk$, assuming PRFs exist, we have $\Pr[\adv \text{ wins } \mathtt{BPGamePlus}] \leq \Pr[\adv \text{ wins } \texttt{H1Game}] + \negl[\lambda] = \varepsilon + \negl[\lambda]$. Therefore, \betpasscbf{} remains $(n, t, \negl)$-secure under \texttt{BPGamePlus} for any $t \in \bigO{\poly[n, \lambda]}$.

Let $\mathbf{C}$ denote the \betpassbloomname{} construction with two backup \betpasscbf{}s $\mathbf{B}_A$ and $\mathbf{B}_B$. $\mathbf{C}$'s completeness follows from $\mathbf{B}_{B}$'s completeness and the fact that any $x$ such that $\mathscr{L}(x) < \tau$ is declared to be not in $S$ by $\mathbf{B}_{2}$ (which is the query algorithm of the \betpassbloomname{}) if and only if the query algorithm of $\mathbf{B}_{B}$ also outputs $0$. $\mathbf{C}$'s soundness follows directly from the soundness of $\mathbf{B}_{A}$ and $\mathbf{B}_{B}$. $\mathbf{C}$'s learning model soundness follows from $\mathscr{L}$'s soundness. Hence, $\mathbf{C}$ is a correct $(n, \tau, \varepsilon, \varepsilon_{p}, \varepsilon_{n})$-LBF. 

Now, we show that the security of \betpassbloomname{} under \texttt{LBPGame} is reducible to the security of \betpasscbf{} under \texttt{BPGamePlus}. Let $E_A$ and $E_B$ be the events that the query goes through backup CBF $\mathbf{B}_A$ and $\mathbf{B}_B$, respectively. Based on the construction of \betpassbloomname{}, $E_A$ and $E_B$ are mutually exclusive events and that $\Pr[E_A \cup E_B] = 1$. Let $C_{T}$ be the overall adversary profit. We denote by $C_{A}$ and $C_{B}$ the expected adversary profit from queries that go to backup CBFs $\mathbf{B}_A$ and $\mathbf{B}_B$, respectively. Based on that, we have:
\[\expect{C_T} = \expect{C_A \mid E_A} \Pr[E_A]  + \expect{C_B \mid E_B} \Pr[E_B] \leq \expect{C_A \mid E_A} + \expect{C_B \mid E_B}\] 

If the total profit $C_T$ is non-negligible, it must be true that either $\expect{C_A \mid E_A}$ or $\expect{C_B \mid E_B}$ is non-negligible. However, since $\mathbf{B}_A$ and $\mathbf{B}_B$ are \betpasscbf{}s, by Theorem~\ref{thm:noybloomtheorem} and our result regarding \texttt{BPGamePlus} above, we know that $\expect{C_A | E_A}$ and $\expect{C_B | E_B}$ are negligible. Therefore, $\expect{C_T}$ is negligible meaning that \betpassbloomname{} is $(n, t, \negl)$-secure under \texttt{LBPGame}.\qed

\end{proof}

The probability that the guess $x^{*}$ that an adversary $\adv$ outputs be a false positive in \betpassbloomname{} is upper-bounded by the decision path where $\adv$ always chooses to bet, i.e, $b = 1$. In this decision path, the false positive probability of \texttt{LBPGame} can be analyzed in a simialr way as done for \texttt{LFAGame}.

\begin{theorem}\label{thm:bet_pass_fpr_expression}
    In \texttt{LBPGame}, for a $\ppt$ adversary $\adv$ that outputs a guess $x^{*}$, the probability that $x^{*}$ is a false positive in a \betpassbloomname{} $\textbf{C}$ storing set $S$ with model $\mathscr{L}$ is 
    \(
        \leq \max(\Pr[\text{FP}(x^{*},S_{A}, m_{A})], \Pr[\text{FP}(x^{*}, S_{B}, m_{B})])
    \) where $S_{A} = \{ x \in S \mid \mathscr{L}(x) \geq \tau\}$, $S_{B} = S \setminus S_{A} = \{ x \in S \mid \mathscr{L}(x) < \tau\}$, $\text{FP}$ is the event denoting a false positive in \securecbf{}, $m_T$, $m_A$, and $m_B$ are the total memory of $\textbf{C}$, memory used by backup CBF $\mathbf{B}_A$, and memory used by backup CBF $\mathbf{B}_B$, respectively, and the probability is taken over the random coins of $\adv$ and $\textbf{C}$.
\end{theorem}

\begin{proof}
    Since we assume $\adv$ always bets and never passes, it always outputs $x^{*}$ to be tested whether it is a FP. A case analysis identical to Theorem~\ref{thm:fullyadaptivefprexpression} shows that the probability of $x^{*}$ inducing a FP is either $\Pr[\text{FP}(x^{*},S_{A}, m_{A})]$ or $\Pr[\text{FP}(x^{*}, S_{B}, m_{B})]$. The upper bound follows from the fact that $x^{*}$ going through backup CBF $\textbf{B}_{A}$ and backup CBF $\textbf{B}_{B}$ are mutually exclusive events.\qed
\end{proof}

%% file: tex/perf-evaluation.tex
\section{Evaluation}
We focus on evaluating FPR vs. memory tradeoffs of our LBF constructions in comparison with known secure CBF constructions. In Section~\ref{sec:numericalanalysis} we conduct a numerical analysis for that in \texttt{LPAGame} and \texttt{LFAGame} based on our FPR model covering a large number of parameters. Note that only  \bloomname{} is included in this analysis, as proving the security of \betpassbloomname{} under \texttt{LPAGame} and \texttt{LFAGame} is left to future work. To emphasize the practicality of our  constructions, Section~\ref{sec:use-case-evaluation} evaluates  \bloomname{} and \betpassbloomname{}, and compares them with Naor et al.'s CBF constructions \securecbf{} and \betpasscbf{}, for a common use-case in a non-adversarial setting.

\subsection{Numerical Analysis}\label{sec:numericalanalysis}

We show performance results for real-world scenarios in both the fully ($\alpha = 1$) and partially ($\alpha \leq 1$) adaptive adversarial models. The goal is to demonstrate scenarios where using a \bloomname{} instead of a \securecbf{} yields better FPR, under a given memory budget, while maintaining security guarantees. Broder et al.~\cite{BroderMitzenmacher2004} show that the probability of $x \in \domain$ being a FP in a CBF with $n_{b}$ bits storing a set $S$ using $n_{h}$ hash functions, is $\Pr[\textrm{FP}(x, S, n_{b})] = (1 - e^{{- n_h |S|}/{n_{b}}})^{n_{h}}$. We choose the number of hash functions $n_{h}$ to be optimal $n_{h} = \ln 2 \cdot (n_{b}/|S|)$, as derived in~\cite{BroderMitzenmacher2004}. 

Similar to~\cite{learnedindexstructures,learnedbloomsandwiching}, we let the false positive probability of a learning model, $\Pr[\textrm{FP}_{L}]$, can be modeled as a fraction of that of a CBF storing set $S$ for the same memory budget (i.e., the learning model has a better FPR than a CBF): $\Pr[\textrm{FP}_{L}(x, S, \mathscr{T}, m)] = c \, (1 - e^{-n_{h}|S|/n_{b}})^{n_{h}}$ where $c \leq 1$. Learning models have both FPs and TNs, and we note that the probability of an entry being a TN in the original set is constant as we assume the input set $S$ does not change after construction. Let $Q_{N}$ be the fraction of TN non-adversarial queries (where the number of adversarial queries is $\alpha N$ and so the number of non-adversarial queries is $|\domain| - \alpha N$). Thus, we have 
\[\Pr[\textrm{TN}_{L}(S, \mathscr{T}, n_{b})] = (1 - \Pr[\textrm{FP}_{L}(S, \mathscr{T}, n_{b})]) \, Q_{N} = (1 - c (1 - e^{- n_{h}|S|/n_{b}})^{n_{h}}) \, Q_{N}\] 

We choose realistic values for our example from prior work on evaluating LBFs~\cite{learnedindexstructures} on Google's Transparency Report. We pick $2$ Megabytes as our memory budget, $m$, chosen from the range of values in Figure 10 of~\cite{learnedindexstructures}. We choose the cardinality of the stored set, $|S|$, as $1.7$ million based on the number of unique URLs in Google's Transparency Report. Kraska et al.~\cite{learnedindexstructures} demonstrates that an LBF with a memory budget of $2$ Megabytes has $0.25$ of the false positive rate of a CBF. Hence, we also use $0.25$ as the value for $c$ in our calculations. Following prior work~\cite{moni1}, we use $128$ bits as the size of our security parameter, $\lambda$. For the case of a \bloomname{}, we let the learning model take $1$ Megabytes while dividing the remaining $1$ Megabytes equally between backup CBFs $\mathbf{B}_{A}$ and $\mathbf{B}_{B}$. The backup CBFs store $S_{A}$ and $S_{B}$, respectively, in \bloomname{}. Our chosen values are summarized in Table~\ref{tab:realisticexperiment}. The full numerical analysis for our model uses $494$ lines of \texttt{C} code and it is available at~\cite{Bodega}.\medskip

\begin{table}[t!]
    \centering
    \caption{Model parameters for \bloomname{} and \securecbf{}.}
    \label{tab:realisticexperiment} 
    \begin{tabularx}{\columnwidth}{@{} l X r @{}} 
        \toprule
        Parameter &  & Value \\
        \midrule
    $m_{T}$ & Total memory budget & 2 MB\\
    $m_{L}$ & Memory budget for the learning model & 1 MB\\
    $m_{A}$ & Memory budget for backup CBF $\mathbf{B}_{A}$ & $0.5$ MB\\
    $m_{B}$ & Memory budget for backup CBF $\mathbf{B}_{B}$ & $0.5$ MB\\
    n & Cardinality of stored set $S \subseteq \domain$ & 1.7 Million \\
    c & $\frac{\Pr[\text{FP}_{L}]\text{ of learning model}}{\Pr[\text{FP}]\text{ of CBF}}$ for same memory budget & 0.25\\
    $\lambda$ & Security parameter & 128 bits\\
    $Q_{N}$ & Fraction of true negative non-adversarial queries & 0.5\\
        \bottomrule
    \end{tabularx}
\end{table}

\noindent\textbf{Varying the fraction of adversarial queries.} 
We take the fraction of adversarial queries $\alpha$ to be a variable ranging from $0$ to $1$. We assume a constant \textit{adversarial strategy}, i.e., the fraction of adversarial queries that are FPs (so they go through backup CBF $\textbf{B}_{A}$) vs. the fraction of adversarial queries that are TNs (so they go through backup CBF $\textbf{B}_{B}$) is constant. In particular, we assume that adversary $\adv$ equally divides its queries between FPs and TNs, so $\frac{\alpha_{P}}{\alpha} = 0.5$ and $Q_{N}$ is $0.5$. As shown in Figure~\ref{fig:inneranalysisfirst}, we observe that a \bloomname{} outperforms a \securecbf{} for the same memory budget when the fraction of adversarial queries is less than a certain cutoff of $0.5$. So, as long as adversarial traffic is at most half of the total workload of an application, \bloomname{} will offer a lower FPR than \securecbf{}. Note that for all our figures, the result for \texttt{LFAGame} is simply the $\alpha = 1$ point in the figure, whereas the entire spectrum of $\alpha$ values shows how FPR varies in the weaker \texttt{LPAGame}.\medskip

\noindent\textbf{Varying the adversarial strategy.} In Figure~\ref{fig:inneranalysissecond}, we relax the assumption that the adversary divides their queries equally between FPs and TNs. We show results for all partitions of $\alpha$ between $\alpha_{P}$ and $\alpha_{N}$. To see how the FPR of \bloomname{} is impacted, we vary the fraction of $\alpha$ assigned to $\alpha_{P}$ from $0$ to $1$. Here, $0$ means all $\alpha N$ adversarial queries are TNs, and $1$ means that all $\alpha N$ adversarial queries are FPs. In this framework, recall that our first calculation (Figure~\ref{fig:inneranalysisfirst}) sets this fraction to $0.5$. The key observation here is that as $\adv$ uses more of its query budget directing traffic to the backup CBF that has the higher FPR, the performance of \bloomname{} degrades. Recall that the upper bound for the FPR of an adversarial query is the FPR of the ``weaker'' backup CBF. For brevity, Figure~\ref{fig:inneranalysissecond} only shows results when $\alpha = 0.2$, i.e., $20\%$ of the workload is adversarial. We note that we observed the same trend for other values of $\alpha$, and hence, we do not include detailed results for that.\medskip

\noindent\textbf{Varying the dataset.} In addition to Google's Transparency Report, we also show results in Figures~\ref{fig:inneranalysisthird} and~\ref{fig:inneranalysisfourth}) for two other datasets evaluated in prior work on LBFs~\cite{sato_matsui,plbf,adabf}. These two datasets are: Malicious URLs Dataset~\cite{maliciousurlsdataset} that contains $223,088$ malicious and $428,118$ benign URLs, and EMBER Dataset~\cite{emberdataset} that
contains $300,000$ malicious and $400,000$ benign files. We change the set's cardinality values, $n$, in Table~\ref{tab:realisticexperiment}. We use the same values listed in the table for all other model parameters. The figures show that, similar to Google's Transparency Report, also for these datasets FPR increases as the fraction of adversarial queries $\alpha$ increases.\medskip

    \begin{figure*}[t!]
        \centering
        \begin{subfigure}[b]{0.49\textwidth}
            \centering
            \includegraphics[width=\textwidth]{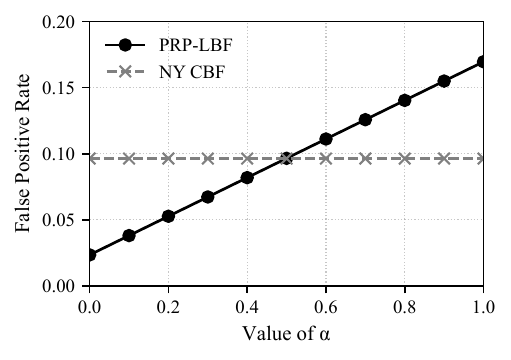}
            \caption[]%
            {{\small Results for the Google Transparency Report as $\alpha$ varies with adversarial queries divided equally between FPs and TNs: $\alpha_{P} = \alpha_{N}$.}}
            \label{fig:inneranalysisfirst}
        \end{subfigure}
        \hfill
        \begin{subfigure}[b]{0.49\textwidth}  
            \centering 
            \includegraphics[width=\textwidth]{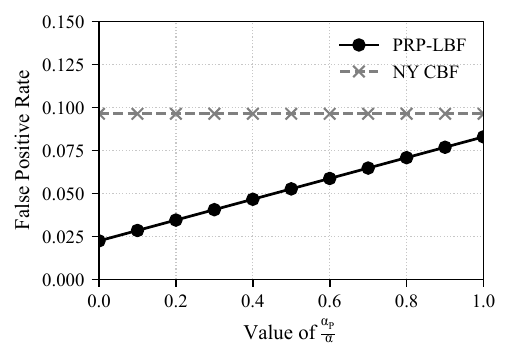}
            \caption[]%
            {{\small Results for the Google Transparency Report with $\alpha = 0.2$ while varying the division of adversarial queries between FPs and TNs (i.e., $\alpha_{P}/\alpha$).}}\label{fig:inneranalysissecond}

        \end{subfigure}
        \vskip\baselineskip
        \begin{subfigure}[b]{0.49\textwidth}   
            \centering 
            \includegraphics[width=\textwidth]{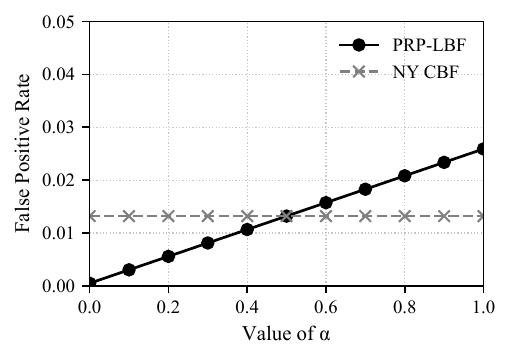}
            \caption[]%
            {{\small Results for the Malicious URLs Dataset (same setting as Figure~\ref{fig:inneranalysisfirst}).}}\label{fig:inneranalysisthird}    
        \end{subfigure}
        \hfill
        \begin{subfigure}[b]{0.49\textwidth}   
            \centering 
            \includegraphics[width=\textwidth]{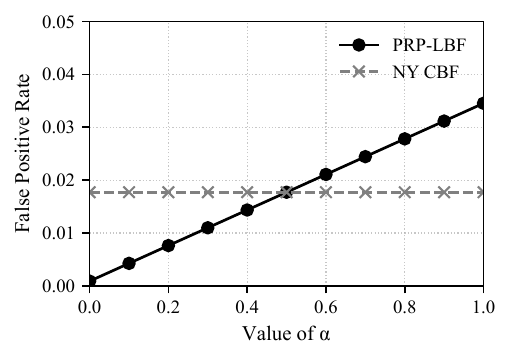}
            \caption[]%
            {{\small Results for the EMBER Dataset (same setting as Figure~\ref{fig:inneranalysissecond}).}}\label{fig:inneranalysisfourth} 
        \end{subfigure}
        \caption{FPR comparison between \bloomname{} and \securecbf{} while varying the fraction of adversarial queries $\alpha$ and the adversarial strategy $\alpha_{P}/\alpha$ for various datasets. Results for the fully adaptive model are the $\alpha = 1$ points in the figures, whereas the entire spectrum of $\alpha$ values shows the results for the partially adaptive model.} 
        \label{fig:analysis_experiments_main}
    \end{figure*}

    \begin{figure*}[ht!]
        \centering
        \begin{subfigure}[b]{0.49\textwidth}
            \centering
            \includegraphics[width=\textwidth]{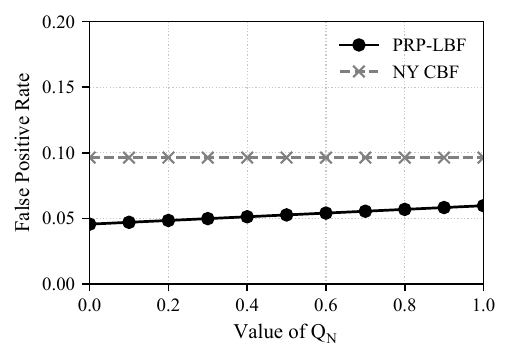}
            \caption[]%
            {{\small $\alpha=0.2$}}    
        \end{subfigure}
        \hfill
        \begin{subfigure}[b]{0.49\textwidth}  
            \centering 
            \includegraphics[width=\textwidth]{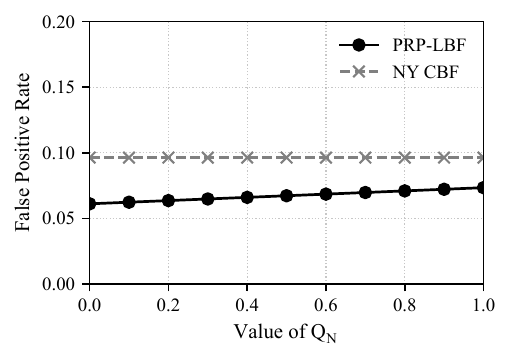}
            \caption[]%
            {{\small $\alpha = 0.3$}}    
        \end{subfigure}
        \vskip\baselineskip
        \begin{subfigure}[b]{0.49\textwidth}   
            \centering 
            \includegraphics[width=\textwidth]{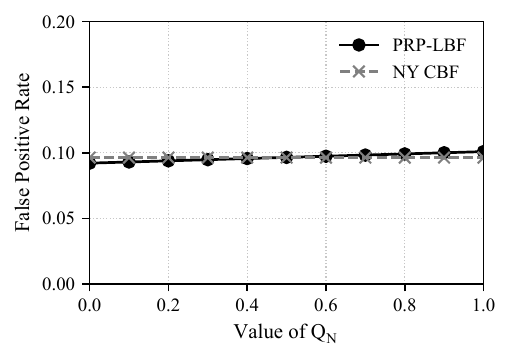}
            \caption[]%
            {{\small $\alpha = 0.5$}}    
        \end{subfigure}
        \hfill
        \begin{subfigure}[b]{0.49\textwidth}   
            \centering 
            \includegraphics[width=\textwidth]{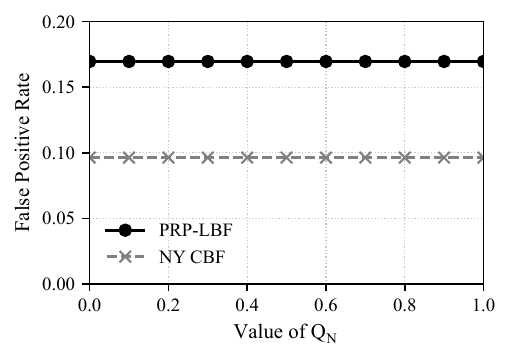}
            \caption[]%
            {{\small $\alpha = 1.0$}}    
        \end{subfigure}
        \caption{FPR comparison between \bloomname{} and \securecbf{} in the partially-adaptive adversarial model for the Google Transparency Report with $Q_{N} \in [0, 1]$.}
        \label{fig:experiment2googletransparency}
    \end{figure*}

\noindent\textbf{Varying the fraction of true negatives.} Recall that in Figure~\ref{fig:inneranalysisfirst}, we set the fraction of non-adversarial queries that are TNs, $Q_{N}$, to be $0.5$. Now we relax this assumption and show results for the entire range of values of $Q_{N} \in [0, 1]$ with $\alpha$ taking $4$ values: $0.2, 0.3, 0.5, 1.0$, such that each value partitioned equally between $\alpha_{P}$ and $\alpha_{N}$. The results are shown in Figure~\ref{fig:experiment2googletransparency}. 

We observe that FPR increases with increasing $Q_{N}$ but the rate of this increase, i.e., $\frac{\partial FPR}{\partial Q_{N}}$, decreases as $\alpha$ increases. This is due to the following reason. In Theorem~\ref{thm:hybriddowntownbodegaexpression}, when $\alpha = \alpha_{P} + \alpha_{N} \to 1$, we have $(1 - \alpha_{P} - \alpha_{N}) \Pr[\text{FP}_{LBF}(x_{i}, S, \mathscr{T}, m_{T})] \to 0$, so the overall FPR of the mixed workload is dominated by the FPR of the backup CBFs which does not depend on $Q_{N}$. On the other hand, when $\alpha \to 0$, the overall FPR of the mixed workload is dominated by $\text{FP}_{LBF}$ which increases with increasing $Q_{N}$ (see Theorem~\ref{thm:fprdowntownbodegafilter}).

\subsection{Use-case Evaluation}\label{sec:use-case-evaluation}

We evaluate the performance of our constructions within the context of a use case to get a sense of how they would perform in practice. Historically, web browsers, including Google Chrome, used a Bloom filter~\cite{gerbetsecurity} to store a set of Malicious URLs. In this design, whenever a user attempts to access a URL on the web browser, the browser queries the Bloom filter for the URL. If the Bloom filter says that the URL is in the Malicious URLs set, the web browser warns the user that they are accessing a potentially unsafe website. The malicious URLs use case has been studied by prior works on LBFs~\cite{sato_matsui,plbf}. Thus, we evaluate this use case using the same public Malicious URLs dataset~\cite{maliciousurlsdataset} as prior work. This dataset contains around 223K malicious and around 428K benign URLs.\medskip 

\noindent\textbf{Implementation and experimental setup.} We implemented \bloomname{}, \betpassbloomname{}, \securecbf{}, and \betpasscbf{} in $903$ lines of Python 3 code, which can be found in our open-source repository~\cite{Bodega}. The implementation allows any PRP/PRF implementation to be plugged in for internal use. Similar to~\cite{moni1}, we use AES to instantiate PRPs and PRFs in our implementation. In particular, we use AES in the ECB mode where the input size of the PRP/PRF is 128 bits (so one block for AES encryption). Our implementation uses the \texttt{PyCryptoDome}~\cite{PyCryptoDome} library for these cryptographic primitives. Our implementation is also modular in a way that allows any machine learning model to be easily plugged in. We tested the correctness of our implementation on a broad range of common classification models, including the Random Forest model, Gaussian Naive Bayes, the Gradient Boosting Classifier, Support Vector Machine-based Classifiers, and Adaptive Boosting, using implementations provided by \texttt{scikit-learn}.

For our experiments, we use a set of $20$ features to train the learning models, including URL length, whether the URL contains an IP address, whether the URL uses a shortening service, whether the URL is ``abnormal'', digit count and letter count of the URL, and whether the URL contains special symbols. This set of features for the Malicious URLs Dataset is common in open-source learning models, and similar features have been used by prior work~\cite{plbf}. In the case of Cuckoo filter-based constructions, we use fingerprints of size $4$ bits, $2$ as the table size constant factor, and $5000$ maximum eviction attempts in the Cuckoo hashing tables. For Cuckoo filter-based constructions, we also skip elements that cannot be inserted after the maximum eviction attempts have been reached.

\begin{figure*}[t!]
    \centering
    \begin{subfigure}[b]{0.49\textwidth}
        \centering
        \includegraphics[width=\textwidth]{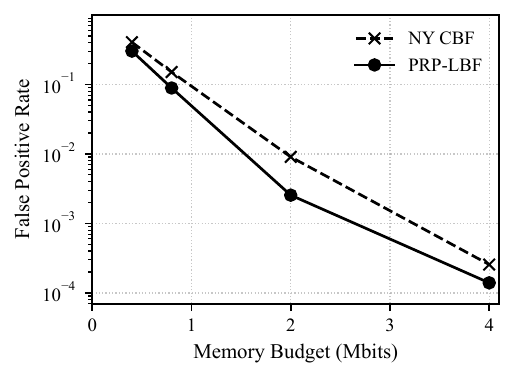}
        \caption[]{\small Gaussian Naive-Bayes Classifier}
        \label{fig:implexperiment1}
    \end{subfigure}
    \hfill
    \begin{subfigure}[b]{0.49\textwidth}  
        \centering
        \includegraphics[width=\textwidth]{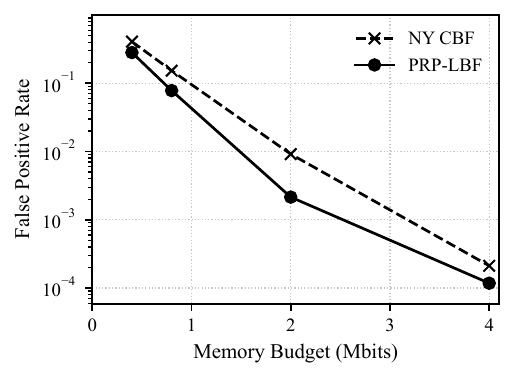}
        \caption[]{\small Linear Support Vector Classifier}
        \label{fig:implexperiment2}   
    \end{subfigure}
        \begin{subfigure}[b]{0.49\textwidth}
        \centering
        \includegraphics[width=\textwidth]{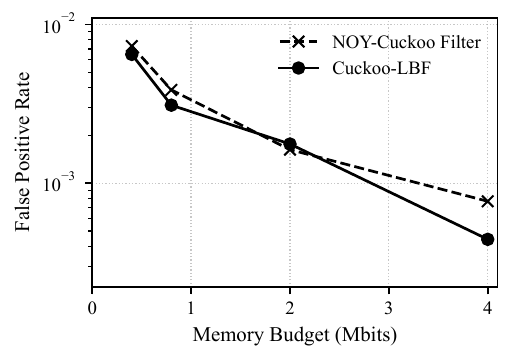}
        \caption[]{\small Gaussian Naive-Bayes Classifier}
        \label{fig:implexperiment3}
    \end{subfigure}
    \hfill
    \begin{subfigure}[b]{0.49\textwidth}  
        \centering
        \includegraphics[width=\textwidth]{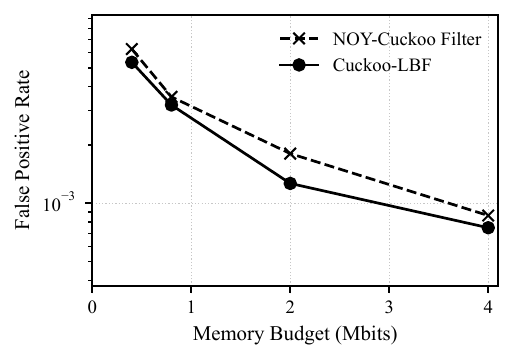}
        \caption[]{\small Linear Support Vector Classifier}
        \label{fig:implexperiment4}   
    \end{subfigure}
    \caption{FPR of our LBF constructions compared to NY-CBF and \betpasscbf{}, with varying memory budget and learning models for the Malicious URLs Dataset.}
    \label{fig:implexperiment}
\end{figure*}

We calculate FPR of all constructions by uniformly randomly sampling $10\%$ of URLs that are \textbf{not} malicious from the Malicious URLs Dataset and counting how many of them are returned as FPs by these constructions. The amount of memory the learning model uses is measured as the serialized size in bytes of the trained classifier, using joblib~\cite{joblib}. After subtracting the memory used by the learning model from the memory budget, we divide the remaining memory budget equally between backup CBFs $\mathbf{B}_{A}$ and $\mathbf{B}_{B}$ in our constructions.\medskip

\noindent\textbf{Results.} Figure~\ref{fig:implexperiment1} shows how FPR varies for \securecbf{} and \bloomname{} as we modify the memory budget. This figure uses the Gaussian Naive-Bayes Classifier as the learning model. Figure~\ref{fig:implexperiment2} shows the same results using a Linear Support Vector Classifier as the learning model. Similarly, Fig.~\ref{fig:implexperiment3} and Figure~\ref{fig:implexperiment4} show how FPR varies for  \betpasscbf{} and \betpassbloomname{} as we modify the memory budget available for the Gaussian Naive-Bayes Classifier and the Linear Support Vector Classifier as the learning model, respectively. 

While our implementations are not focused on optimization, we see a consistent trend of our LBF constructions having lower FPRs than CBF constructions for the same memory budget. This is consistent with prior work~\cite{learnedindexstructures} that shows a similar trend between non-adversarial LBF constructions and non-adversarial CBF constructions. An interesting outlier is the FPR of \betpassbloomname{} being slightly larger than the FPR of \betpasscbf{} in one of the data points of Figure~\ref{fig:implexperiment3}. We conjecture that this is due to our naive method of equally distributing the memory leftover, i.e., after taking out the memory needed for the learning model, between the backup CBFs. We leave investigating better memory allocation strategies and other optimizations to future work.\medskip

\noindent\textbf{Large memory budgets.} We explore how the trend of Figure~\ref{fig:implexperiment} continues as we keep increasing the memory budget to the point where it no longer becomes a bottleneck for FPR. To better understand this, we conduct a second set of experiments over a much larger range of memory budgets with the results shown in Figure~\ref{fig:extendedimplexperiment}. We observe that the trend from Figure~\ref{fig:implexperiment} continues. FPR for both our LBF constructions and the CBF constructions eventually approaches the same value as the memory budget increases. 

\begin{figure*}[t!]
    \centering
    \begin{subfigure}[b]{0.49\textwidth}
        \centering
        \includegraphics[width=\textwidth]{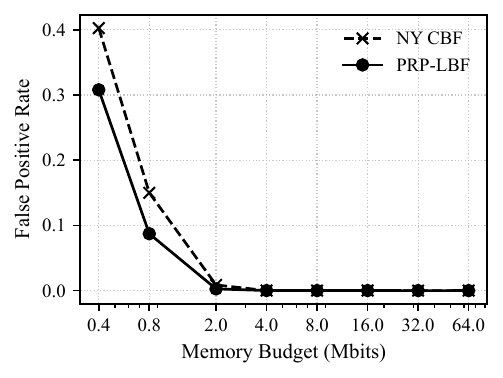}
        \caption[]{\small Gaussian Naive-Bayes Classifier}
        \label{fig:extendedimplexperiment1}
    \end{subfigure}
    \hfill
    \begin{subfigure}[b]{0.49\textwidth}  
        \centering
        \includegraphics[width=\textwidth]{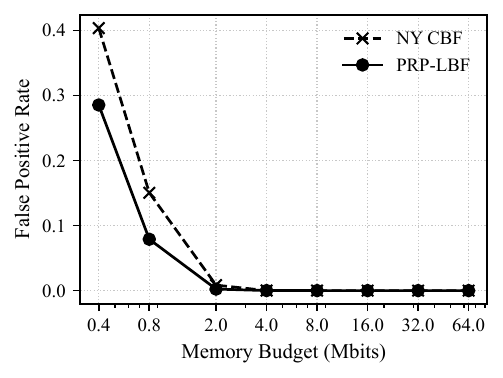}
        \caption[]{\small Linear Support Vector Classifier}
        \label{fig:extendedimplexperiment2}   
    \end{subfigure}
        \begin{subfigure}[b]{0.49\textwidth}
        \centering
        \includegraphics[width=\textwidth]{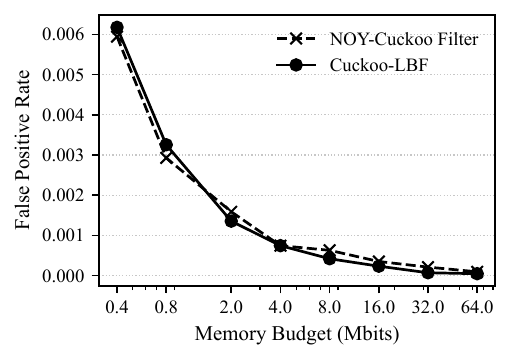}
        \caption[]{\small Gaussian Naive-Bayes Classifier}
        \label{fig:extendedimplexperiment3}
    \end{subfigure}
    \hfill
    \begin{subfigure}[b]{0.49\textwidth}  
        \centering
        \includegraphics[width=\textwidth]{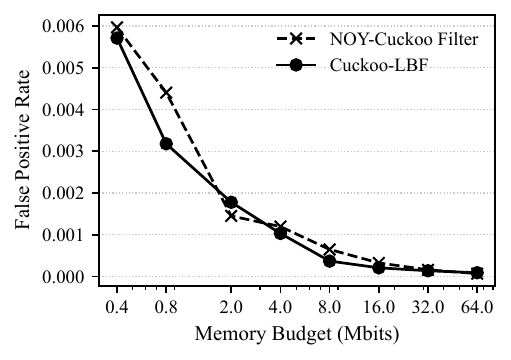}
        \caption[]{\small Linear Support Vector Classifier}
        \label{fig:exte dedimplexperiment4}   
    \end{subfigure}
    \caption{FPR of our LBF constructions compared to NY-CBF and \betpasscbf{} with large memory budgets and different learning models (for the Malicious URLs Dataset).}
    \label{fig:extendedimplexperiment}
\end{figure*}

%% file: tex/acknowledgments.tex
\section*{Acknowledgments}

The work of G.A. is supported by NSF Grant No. CNS-2226932.

%% file: tex/appendices.tex
\subsection*{Generative AI Disclosure}

The paper was written directly by the authors. The student author used ChatGPT by OpenAI and Claude by Anthropic for feedback on technical writing and to edit and debug the Ti\textit{k}Z and matplotlib figure code in the \LaTeX{} manuscript. The artifacts of our work were primarily written and verified by the authors. The student author used Cursor IDE, which includes a generative AI code assistant, for the following tertiary coding tasks: unit-test generation, code refactoring, feedback on bug-fixing, feedback on anonymizing the artifact, and code reviews.